\documentclass[11pt]{article}

\usepackage[margin=1.15in]{geometry}
\usepackage{amsmath,amssymb,amsthm,mathtools}
\usepackage{enumitem}
\usepackage{hyperref}
 \usepackage{colortbl}
 \usepackage{tikz}

\newtheorem{theorem}{Theorem}[section]
\newtheorem{proposition}[theorem]{Proposition}
\newtheorem{corollary}[theorem]{Corollary}
\newtheorem{lemma}[theorem]{Lemma}
\newtheorem{definition}[theorem]{Definition}
\newtheorem{remark}[theorem]{Remark}
\newtheorem{example}[theorem]{Example}

\newcommand{\up}{\uparrow}
\newcommand{\cmark}{\mathrm c}
\newcommand{\N}{\mathbb N}
\newcommand{\Nzero}{\mathbb N_0}
\newcommand{\Nb}{\widehat{\mathbb N}_0}
\newcommand{\Pow}{\mathcal P}
\newcommand{\F}{\mathcal F}
\newcommand{\G}{\mathcal G}

\newcommand{\W}{\mathcal W}
\newcommand{\cS}{\mathcal S}
\newcommand{\SCF}{\mbox{\boldmath $\Phi$}}
\newcommand{\siga}{\sigma}

\newcommand{\co}{\mathrm{co}}

\DeclareRobustCommand{\rchi}{{\mathpalette\irchi\relax}}
\newcommand{\irchi}[2]{\raisebox{\depth}{$#1\chi$}} 

\title{
{\color{black} Coalition strategy-proof} anonymous binary social choice in a countably infinite society
\\[0.4em]
}
\author{Achille Basile\thanks{Corresponding author, Dipartimento di Scienze Economiche e Statistiche,  Universit\`a Federico II,
80126 Napoli, Italy, E-mail: basile@unina.it;  
},\,
{K. P. S. Bhaskara Rao\thanks{Department of Computer Information Systems, Indiana 
University Northwest, Gary, IN 46408, E-mail: bkoppart@iu.edu}},\,
{Surekha Rao\thanks{School of Business and Economics, Indiana University Northwest,
Gary, IN 46408, E-mail: skrao@iu.edu}} ,\,
}

\begin{document}

\maketitle

\begin{abstract}

{\color{black}
We study anonymous binary social choice functions on the universal weak-preference
domain when the society of voters is countably infinite. Voters may strictly support
one of the two alternatives or be indifferent, and non-manipulability is required in
the form of coalitional strategy-proofness. Full anonymity reduces the relevant
information about a coalition not to its cardinality alone, but to its cardinality
signature: finite of size \(k\), infinite with infinite complement, or cofinite with
complement of size \(k\). Recording the signatures of the two strict-support
coalitions yields an infinite triangular grid, ordered by increasing support
for one alternative and decreasing support for the other. We prove that anonymous
coalition strategy-proof binary social choice functions are exactly those obtained from
super order closed subsets, i.e. upper sets, of this grid. We then give a geometric classification of
these subsets. Unlike in the finite-voter case, they need not be generated by their minimal
elements: countable infinity creates missing-boundary regions corresponding to
finite opposition or cofinite strict support. We prove an exact representation of these subsets by genuine or missing-boundary corners, and derive a
nonredundant parametrization by an antichain and, when necessary, one exceptional
missing-boundary region. This yields four generalized zigzag forms and completes the
infinite-society counterpart of the known geometric representation for
finite electorates.}

\end{abstract}




\bigskip

 \newpage

\section{Introduction}

{\color{black}
The present paper is concerned with anonymous binary social choice functions. The voters have to choose collectively between two alternatives, say $a$ and $b$, and they are allowed to express indifference. Indifference records the absence of a strict preference, a possibility relevant to many binary collective decisions, including referenda, ratification votes, and committee decisions.\footnote{Indifference is a preference relation, whereas abstention is an action that may also be chosen by a voter with a strict preference. If abstention is encoded as a report of indifference, it supplies the same input to the rule, while the assessment of manipulation remains based on the voter's true preference.}
}
In the case of finitely many voters,
Basile, Rao and Bhaskara Rao \cite{mss} show that anonymous and strategy-proof binary
social choice functions admit a simple geometric representation. 

If the number of
voters is \(n\), anonymity reduces each profile to the pair
$\,
(k,m),
\,$
where \(k\) is the number of voters who strictly prefer \(a\) to \(b\), and \(m\) is the
number of voters who strictly prefer \(b\) to \(a\). Thus the relevant domain becomes
the finite triangular grid
\[
G_n=\{(k,m)\in \mathbb N_0^2:k+m\le n\}.
\]
Adopting the partial ordering according to which $(k,m)\leq_{G_n} (k',m')$ if and only if $k\le k'$ and $m\ge m'$,
strategy-proofness is then represented by a monotonicity condition on this grid, and
the admissible social choice functions can be visualized through the shape of the
corresponding regions of \(G_n\). Indeed  such regions have to be separated by a vertical-horizontal zigzag line whose
sequence of corners is characterized by the property of forming an antichain of $(G_n, \leq_{G_n})$.

Our aim is to extend this analysis to a countably infinite \lq\lq society\rq\rq\, that we represent by the set \(\mathbb N=\{1, 2, \dots \}\) of natural numbers. 
Thus each voter  $i\in \mathbb N$ may report one of the three possibilities: strict preference for \(a\), strict preference for \(b\), or indifference.
We restrict attention to anonymous social choice functions, where anonymity means
invariance under all permutations of \(\mathbb N\). 
{\color{black}
In an infinite society, invariance under all permutations is a substantive version of anonymity rather than the only possible extension of the finite notion. Weaker notions, such as finite and bounded anonymity, have also been considered. Notably, for the purpose of obtaining density-based analogues of majority rule,  Fey \cite{Fey2002} adopts  the intermediate notion of bounded anonymity.
 We  adopt full permutation invariance to express symmetry among voters without assigning a substantive role to their enumeration. At the level of a single coalition, the information preserved by this symmetry is its cardinality together with the cardinality of its complement. Thus finite coalition sizes and finite complement sizes remain relevant, while all infinite non-cofinite coalitions are equivalent under relabelling. These distinctions are recorded by the cardinality signatures introduced below.}
 
The use of \(\mathbb N\) as the set of decision-making units  is an idealized large-society model, well represented in the literature on  social choice and voting 
{\color{black}
as well as the neighbouring Arrovian tradition. In addition to the aforementioned Fey \cite{Fey2002}, 
Cato, R\'emila and Solal \cite{CatoRemilaSolal2021}, for instance, motivate countably infinite voting
models through intergenerational choice (each successive generation adds an additional voter to a sequence, thereby rendering the relevant electorate countably infinite),  and other interpretations involving
infinite streams of decision-making units; Pivato and Fleurbaey \cite{PivatoFleurbaey2024}
survey recent developments in infinite-population ethics and intergenerational
social welfare. Related measure-theoretic approaches also show that the structure assigned to an infinite population affects the resulting social-choice possibilities; see, for instance, Cato \cite{Cato2021JME} on preference aggregation with a finitely additive measure space of agents and the role of atoms in Arrow-type compatibility results.
Cato \cite{Cato2022} discusses infinite-population
preference aggregation in connection with both classical impossibility theory and
coalitional strategy-proofness.
 His framework requires at least three alternatives and his central concern is the compatibility of a stability axiom\footnote{Adjoining the social outcome itself to the profile as an additional voter leaves the outcome unchanged.} with coalitional/individual strategy-proofness and non-dictatorship, without considering anonymity.
Differently, the present paper fixes exactly two alternatives and anonymity throughout, imposes coalitional strategy-proofness without any stability requirement, and asks for an exhaustive geometric classification of the admissible rules.
}

 More generally, infinite-population models have
been used to clarify how standard social choice principles behave once the
finiteness of the population is removed.

{\color{black}
This qualitative shift is likely best known on the Arrovian side of the theory, where stepping from a finite to an infinite  society,  the celebrated Arrow's impossibility theorem becomes a possibility one: an infinite electorate admits non-dictatorial rules governed by a free ultrafilter -- "invisible dictators" in the terminology of Kirman and Sondermann \cite{KS} (see also Fishburn \cite{Fishburn1970}; Armstrong \cite{TA}; and, for the countable case specifically, Mihara \cite{M0}). Since Gibbard--Satterthwaite theorem on the manipulability of voting is  known to be intimately connected  to Arrow's theorem (see  Reny \cite{Reny}, for example), the same qualitative shift is to be expected -- and indeed occurs -- on the strategy-proofness side. In this setting, Mihara \cite{M1} shows that coalition strategy-proof social choice functions depend only on voters' most-preferred alternatives, and Rao, Basile and Bhaskara Rao \cite{etb} sharpen this into an explicit ultrafilter representation of these functions.  The present paper documents an analogous phenomenon for coalition strategy-proof, binary, anonymous, indifference-admitting choice: the geometric classification of Section \ref{sei} includes shapes with no finite-society counterpart (Corollary \ref{thm:zigzag}, notably super order closed\footnote{We recall that in the framework of partially ordered sets,  super order closed subsets are also called upper sets.} sets that cannot be generated by an antichain, Example \ref{62}).
}

There are two reasons why the passage from finitely many voters to countably many
voters is not merely formal. 
The first one is strategic. With infinitely many voters,
the natural non-manipulability notion is  coalitional  strategy-proofness, also called
weak group strategy-proofness. This is the terminology adopted in Barber\`a, Berga
and Moreno \cite{barbera2012} and in Basile, Bhaskara Rao and Rao \cite{scw}. 
{\color{black}
Group incentive-compatibility notions for social choice functions remain an active topic; for example, Arribillaga, Mass\`o and Neme \cite{ArribillagaMassoNeme2026} study group obvious strategy-proofness, an extension of obvious strategy-proofness to coalitions.  

We recall that
on the universal domain, individual and coalitional strategy-proofness coincide when the electorate is finite (Dasgupta, Hammond and Maskin, \cite{Dasgupta}); Barbera, Berga and Moreno \cite{BBM10} and Le Breton and Zaporozhets \cite{lebreton} show, more generally, that whether the equivalence holds true,  is a question of domain restriction, not of population size as such. The gap we rely on here is of the second, population-driven kind: with an infinite electorate, coalitional strategy-proofness is strictly stronger than individual strategy-proofness, even on the universal domain, and \cite[Example 2.2]{et} exhibits this gap directly within the present binary framework.
 The distinction acquires normative weight in our setting: every coalition strategy-proof social choice function is weakly Pareto optimal relative to its range: no alternative in that range is strictly preferred by every voter to the selected outcome.
 With an infinite set of voters individual strategy-proofness alone does not suffice for this property \cite[Example 2.2]{et}. 
 If weak Pareto optimality is regarded as a desirable requirement of a collective choice rule, this distinction provides an additional reason for adopting coalitional -- not merely individual -- strategy-proofness once the electorate is infinite.
 {\color{black}Finally, coalitional strategy-proofness also rules out profitable deviations requiring changes in infinitely many reports. Such deviations can alter whether a strict-support coalition is finite, infinite with infinite complement, or cofinite, whereas no finite number of report changes can do so. This is a source of strategic restrictions that individual strategy-proofness alone does not impose. Their role in the genuinely infinitary features of the classification of Section \ref{sei} is made more explicit in the concluding remarks (Section \ref{conclusioni}).}
}
The second reason is geometric.
In the finite case, anonymity leaves only finite cardinalities to be counted. In a countably infinite society, ordinary finite counts no longer exhaust the anonymous information contained in a coalition of voters. A subset of \(\mathbb N\), up to permutation of the voters, may be finite of cardinality \(k\), infinite with infinite complement, or cofinite with complement of cardinality \(k\).

This observation leads to the notion of cardinality signature. We introduce a linearly
ordered set
\[
T=\mathbb N_0\cup\{\infty\}\cup\{k^c:k\in\mathbb N_0\},
\]
where \(k\in \N_0\) records finite cardinality \(k\), \(\infty\) records the infinite non-cofinite
case, and \(k^c\) records the cofinite case with complement of cardinality \(k\).
The signature of a set \(E\subseteq\mathbb N\), denoted by \(\sigma(E)\), records exactly
the permutation-invariant cardinality information of \(E\).

In the weak-preference model the two relevant coalitions associated with a profile
\(P\) of preferences are
\[
D(a,P)=\{i\in\mathbb N: a \succ_{P_i} b\},
\qquad
D(b,P)=\{i\in\mathbb N:b\succ_{P_i} a\}.
\]
The voters outside these two sets are indifferent. 
{\color{black}
The pair
\[
\rchi_w(P)=\bigl(\sigma(D(a,P)),\sigma(D(b,P))\bigr).
\]
records the cardinality-signature information associated with the two strict-support coalitions. As will follow from the representation theorem, this is exactly the ordered cardinality information relevant for anonymous coalition strategy-proof rules.

Notice that  $\rchi_w$ need not distinguish all profiles that are inequivalent under permutations, since it does not separately record the cardinality of the indifferent class.
}

The range of the map $P\mapsto \rchi_w(P)$ is an infinite triangular grid
\[
G=\rchi_w(\mathcal P)\subseteq T\times T.
\]
This grid is the countably infinite analogue of the finite grid \(G_n\) used in \cite{mss}.

The character approach of Basile, Bhaskara Rao and Rao \cite{scw} is the main methodological
tool used here. That paper develops canonical functional forms for non-manipulable
two-valued social choice functions by representing them through super order closed
subsets of suitable partially ordered sets. 
{\color{black}The present analysis also belongs to a line of work on two-valued coalition strategy-proof social choice functions with indifference. In particular, Basile, Rao and Bhaskara Rao \cite{BasileRaoBhaskaraRao2021}  provide a representation formula for coalition strategy-proof social choice functions whose range has cardinality two, on domains allowing indifference and for societies and alternatives of arbitrary cardinality. Differently, the present paper imposes anonymity, fixes a countably infinite society, and shows that these additional restrictions turn the representation into the geometry of super order closed subsets of the cardinality-signature grid $G$. }

Applying the character methodology to the present setting,
the result is that anonymous coalition
strategy-proof binary social choice functions are represented exactly by the super
order closed subsets of the infinite grid \(G\). We recall that super
order closed sets  are also called upper sets.

{\color{black}
The specifically infinite part of the analysis has three interconnected components. First, the anonymous weak-preference domain is represented by the cardinality-signature grid $G$, whose structure reflects the distinction between finite, infinite non-cofinite, and cofinite coalitions. Second, a permutation lemma shows that the order on $G$ is compatible with comparisons of actual coalitions after a suitable relabelling of the voters, which is the key step in the representation theorem. Third, the SOC subsets of $G$ exhibit genuinely infinitary geometric features: they need not be generated by their minimal elements, and their description may require missing-boundary components in addition to genuine corners. Proposition \ref{use of parameters} completes this description with a nonredundant parametrization by an
antichain and an optional exceptional region, subject to explicit admissibility conditions.
Combined with Theorem \ref{53}, it generates every anonymous coalition strategy-proof rule
exactly once.
}


Next, the organization of the sections is as follows. Section \ref{sezione2} introduces the social choice model and the
basic order-theoretic terminology. Section \ref{sezione3} is a set-theoretic preliminary section:
it classifies the anonymous super order closed families of subsets of \(\mathbb N\).
Section \ref{sezione4}
treats the strict preference case. This case is one-dimensional, because
without indifference the set \(D(a,P)\) determines the relevant binary information.
It also prepares the cardinality-signature notation used later. Section \ref{sezione5} turns to  weak preferences. We define the weak-domain character \(\rchi_w\), identify its
range with the infinite grid \(G\), and prove the representation theorem: anonymous
coalition strategy-proof binary social choice functions are in one-to-one
correspondence with the super order closed subsets of \(G\). Having established the representation theorem, which constitutes the first step towards the infinite-voters analogue of \cite{mss}, in Section \ref{sei} we undertake the second step, which is of a geometric nature. We describe the shape of the super order closed
subsets of \(G\).  This is richer than in the finite  case, showing phenomena that  emerge because the set of voters is infinite. The Section is enriched with figures intended to clarify the results obtained. The more technical proofs, the omission of which from the main text makes it flow better and renders it more readable, are in Section \ref{dimostrazioni} preceded by Section \ref{conclusioni} which contains the conclusions.

\section{The social choice model and basic notation}\label{sezione2}

A countably infinite society is meant to be a countable, not finite, set whose elements represent decision making units, or voters. We enumerate the voters as
\[
\N=\{1,2,\ldots\}.
\]
Throughout the paper the fact that $i\in\N$ may have two meanings: voter $i$, or the natural number $i$. No ambiguity will arise.

Let $a,b$ be the  alternatives under discussion. We denote by $\W=\W(\{a,b\})$ the set of all complete and transitive binary relations on $\{a,b\}$. Simply this means that $W\in\W$, is one of the following  relations
\[
 a \succ_W b, \qquad b\succ_W a, \qquad a\sim_W b
\]
with the usual meaning. Preference profiles form the universal weak domain, i.e. the set 
$
\Pow=\W^\N.
$
A profile is denoted by $P=(P_i)_{i\in\N}$.

For a profile $P\in\Pow$ we set
\[
D(a,P)=\{i\in\N: a\succ_{P_i} b\},
\]
\[
D(b,P)=\{i\in\N: b\succ_{P_i} a\},
\]
and
\[
I(P)=\{i\in\N: a\sim_{P_i} b\}.
\]
Thus $D(a,P)$, $D(b,P)$ and $I(P)$ form a partition of $\N$.

A binary social choice function is a map
\[
\phi:\Pow\longrightarrow \{a,b\}.
\]
For the sake of brevity we write scf (scfs, for the plural) to mean \lq\lq binary social choice function\rq\rq. We also use \lq\lq rule\rq\rq\, sometime.

\begin{definition}[Coalitional strategy-proofness]
Let $\phi:\Pow\to\{a,b\}$ be a scf. A  coalition $S\subseteq\N$ strongly manipulates the profile $P$ under $\phi$ if there is $Q\in\Pow$ such that
\begin{enumerate}[label=(\roman*)]
\item $Q_i=P_i$ for every $i\notin S$;
\item every voter $i\in S$ strictly prefers $\phi(Q)$ to $\phi(P)$ according to $P_i$.
\end{enumerate}
The scf $\phi$ is coalition strategy-proof, or CSP, if no coalition strongly manipulates any profile.
\end{definition}

\noindent
Observe that this notion is the {\it weak group strategy-proofness} of Barber\`a et al. \cite{barbera2012}. 

\begin{definition}[Anonymity]
Let $\pi$ be a permutation of $\N$. For a profile $P$, define the profile $P\circ \pi$ by
\[
(P\circ \pi)_i=P_{\pi(i)}.
\]
A scf $\phi:\Pow\to\{a,b\}$ is anonymous if
\[
\phi(P\circ \pi)=\phi(P)
\]
for every $P\in\Pow$ and every permutation $\pi$ of $\N$.
\end{definition}

Observe that
\[
D(a,P\circ \pi)=\pi^{-1}(D(a,P)),
\qquad
D(b,P\circ \pi)=\pi^{-1}(D(b,P)).
\]

\noindent
We recall the basic order-theoretic terminology that will be used throughout the paper. 
If $(L,\geq)$ is a partially ordered set, a subset $C\subseteq L$ is super order closed, briefly SOC, if
\[
x\in C,\quad y\geq x \quad \Longrightarrow\quad y\in C.
\]
Thus, a SOC subset is an upper set of the poset.  A subset $A$ of $L$ is an antichain when its distinct elements are pairwise incomparable: $x, x'\in A, x\neq x'\Rightarrow x'\ge x$ is false. The notation
$\up M$ is used for the SOC subset of $L$ generated by the subset $M$ of $L$: $\up M=\{x\in L: \exists m\in M$ such that $x\ge m\}.$ Of course the empty set is an antichain and it is SOC. The SOC subset of $L$ generated by the empty set is the empty set itself.  Finally we recall that  the antichains  of a finite poset $L$ are in a one-to-one correspondence with its SOC subsets, a circumstance whose failure in the infinite grid $G$ we shall introduce, plays a central role in the paper, (Example \ref{62}, {\it infra}).

\section{Anonymous committees on a countably infinite set}\label{sezione3}

The purpose of this  purely set-theoretic section is to classify the anonymous SOC families of subsets of $\N$. This classification will be used in the strict-preference representation, and it will also motivate the cardinality signatures introduced later.

With symbols such as $\F,\G$, we denote families of subsets of voters. Thus $\F,\G\subseteq 2^\N$. We say that $\F$ is SOC if it is a SOC subset of the poset $(2^\N,\subseteq)$. 

\begin{definition}
A family $\F\subseteq 2^\N$ is anonymous if
\[
E\in\F,\quad \pi\text{ permutation of }\N
\quad\Longrightarrow\quad
\pi(E)\in\F.
\]
The family $\F$ is said to be ANSOC if it is both anonymous and SOC.
\end{definition}

We also use the term {\it committee}  for a SOC family, with one caveat: stricto sensu a committee (also {\it simple game} in the Game Theory literature - \cite{peleg}) is a nonempty family of coalitions (i.e. nonempty sets) of voters which is closed under supersets. So we are allowing  the use of the term committee also for the empty family and the whole power set. 
 Adopting this, we may  replace  {\it  \lq\lq ANSOC $\F$\rq\rq\,} with {\it \lq\lq anonymous committee $\F$\rq\rq}.
\begin{example}
{\rm The following families are ANSOC:
\begin{enumerate}[label=(\roman*)]
\item $\G_0=2^\N$;
\item for every $k\in\N$, 
\[
\G_k=\{E\subseteq\N: |E|\geq k\};
\]
\item 
\[
\G_\infty=\{E\subseteq\N:E\text{ is infinite}\};
\]
\item
\[
\G_{\co-\infty}=\{E\subseteq\N:E\text{ is cofinite}\};
\]
\item for every $k\in\N$,
\[
\G_{\co-k}=\{E\subseteq\N: |E^c|<k\};
\]
\item $\G_{\co-0}=\varnothing$.
\end{enumerate}
They form the chain
\[
\G_0\supset \cdots \supset \G_k\supset \G_{k+1}\supset\cdots
\supset \G_\infty\supset \G_{\co-\infty}
\supset\cdots\supset \G_{\co-(k+1)}\supset \G_{\co-k}\supset\cdots\supset \G_{\co-0}.
\]}
\end{example}

\noindent
In the power set of $\N$ there are no other ANSOC families, as the following theorem establishes.
\begin{theorem}[Classification of ANSOC families]\label{classificazione}
If $\F$ is an ANSOC family of subsets of $\N$, then $\F$ is one of the families listed in the previous example.
\end{theorem}

\begin{proof} See the Appendix. \ref{dim classificazione}
\end{proof}

\begin{remark}{\rm The representation of  CSP scfs, without anonymity requirement,  obtained  in Basile,  Rao and Bhaskara Rao \cite{et} with reference to arbitrary, not necessarily countable (finite or not), set of voters relies on the
 duality  operation on  families $\F$ of subsets of a set. This operation is defined by setting
\[
E\in \F^\circ \quad\Longleftrightarrow\quad E^c\notin \F.
\]
We observe that $\G_k^\circ=\G_{\co-k}$ and $\G_\infty^\circ=\G_{\co-\infty}$. We shall not use this operation in the sequel, but the observation  explains the notation and the symmetry of the chain.}
\end{remark}

\section{The strict-preference case: a one-dimensional representation}\label{sezione4}

{\color{black}
Before considering the weak-preference domain, we use the strict-preference case as a one-dimensional benchmark. It provides a self-contained preliminary classification and introduces the cardinality-signature structure that will be used in the main two-dimensional construction.
 }The reason  
is simple. If voters are not allowed to express indifference, then each voter either supports $a$ or supports $b$. Hence the set $D(a,P)$ determines the entire profile as far as the binary collective choice is concerned. Under anonymity, the relevant information is the permutation-invariant cardinality information of $D(a,P)$. This results in the possibility of a one-dimensional geometrical interpretation of the class of all scfs  which
are anonymous and CSP.

In this section the domain  is a subset $\Pow^s$ of the previous $\Pow$, precisely
$
\Pow^s=\cS^\N
$, where $\cS$ only contains the two strict preference relation 
 $\quad a \succ b,$ \, and $\, b\succ a.$
We are interested in the class, denoted by $\SCF_{AN}^s$,  of all social choice functions  $\phi:\Pow^s\to\{a,b\}$ which are anonymous and CSP.

We recall  the following  representation result whose validity, we point out, is not limited to finite sets of voters, the case where it was first achieved in \cite{LS}. The formulation we present is  from Basile,  Rao and Bhaskara Rao \cite{et}.

\begin{theorem}[Committee representation under strict preferences]\label{41}
Let $V$ be an arbitrary set of voters. On the strict universal domain $S^V$, the CSP social choice functions with range contained in $\{a,b\}$ are all and only the functions
\[
\phi_\F(P)=
\begin{cases}
 a, & \text{if }D(a,P)\in\F,\\
 b, & \text{if }D(a,P)\notin\F,
\end{cases}
\]
where $\F\subseteq 2^V$ is a SOC subset of the poset $(2^V, \subseteq)$.
\end{theorem}
If we consider anonymous functions in  the case where $V=\mathbb N$, we get the following corollary to the previous theorem.

\begin{corollary}[The countably infinite anonymous strict case]\label{prima rappresentazione}
If the voters of a countable infinite society can only express strict preference between two alternatives $a,\, b,\,$
the class $\SCF_{AN}^s$ consists exactly of the functions $\phi_\F$ defined in Theorem \ref{41} where $\F$ runs along the chain
\[
\G_0\supset \cdots \supset \G_k\supset \G_{k+1}\supset\cdots
\supset \G_\infty\supset \G_{\co-\infty}
\supset\cdots\supset \G_{\co-(k+1)}\supset \G_{\co-k}\supset\cdots\supset \G_{\co-0}.
\]
\end{corollary}

\begin{proof}
If $\F$ is ANSOC, then the rule $\phi_\F$ is anonymous and CSP. Conversely, let $\phi$ be anonymous and CSP. By the preceding theorem, $\phi=\phi_\F$ for some SOC family $\F$. We prove that $\F$ is anonymous. Take $E\in\F$ and a permutation $\pi$ of $\N$. Since the domain is universal, choose a strict profile $P$ such that $D(a,P)=E$. Then $\phi(P)=a$. By anonymity, $\phi(P\circ\pi)=a$. Since $D(a,P\circ\pi)=\pi^{-1}(E)$, we get $\pi^{-1}(E)\in\F$. As $\pi$ is arbitrary, $\F$ is anonymous. The classification theorem now applies.
\end{proof}

We now rewrite this result in the language of {\it characters} introduced  in Basile, Bhaskara Rao and Rao \cite{scw}. Let
\[
T=\N_0\cup\{\infty\}\cup\{k^c:k\in\N_0\}.
\]
We endow $T$ with the linear order 
\[
0<_ T1<_T2<_T\cdots<_Tk<_Tk+1<_T\cdots<_T\infty<_T\cdots<_T(k+1)^c<_Tk^c<_T\cdots<_T1^c<_T0^c.
\]
Figure 1 below represents $T$ on a compact interval of the real line. 
Thus $0$ is the minimum and $0^c$ is the maximum.

$$\begin{tikzpicture}[xscale=0.38, yscale=0.40]
\draw [->] [thick] (0,0) -- (43,0);
\node at (20,-1) {$\infty$};  \node at (10,-1){...}; \node at (18.2,-1){...}; \node at (30,-1){...}; \node at (22.3,-1){...};
{\footnotesize 
\node at (0,-1) {$0$};
\node at (1,-1) {$1$};
\node at (2,-1) {$2$};\node at (3,-1) {$3$}; \node at (4,-1) {$4$};
\node at (40,-1) {$0^c$};
\node at (39,-1) {$1^c$};\node at (38,-1) {$2^c$}; \node at (37,-1) {$3^c$};\node at (36,-1) {$4^c$};
\node at (16.5,-1) {$n$}; \node at (25,-1) {$k^c$}; 
\draw[fill] (0,0) circle [radius=0.17]; \draw[fill] (20,0) circle [radius=0.17]; \draw[fill] (40,0) circle [radius=0.17];
\draw[fill] (1,0) circle [radius=0.1]; \draw[fill] (2,0) circle [radius=0.1]; \draw[fill] (3,0) circle [radius=0.1];
\draw[fill] (4,0) circle [radius=0.1]; \draw[fill] (16.5,0) circle [radius=0.1]; \draw[fill] (25,0) circle [radius=0.1];
\draw[fill] (39,0) circle [radius=0.1]; \draw[fill] (38,0) circle [radius=0.1]; \draw[fill] (37,0) circle [radius=0.1];
\draw[fill] (36,0) circle [radius=0.1]; 
}
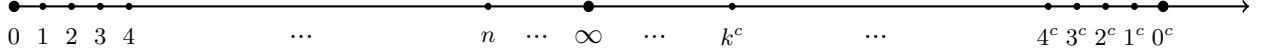
\node at (20,-3) {Figure 1: the linearly ordered set $T$};
\end{tikzpicture}
$$
Although the following remark is elementary, we explicitly record it.

\begin{remark}\label{due}
{\rm In $(T,\le_T)$ every nonempty subset $X$  has $\sup$ and $\inf$. In greater detail:

$\O\neq X\cap [0,\infty] \Rightarrow X$ has minimum.

$\O= X\cap [0,\infty],$ i.e. $ X\subseteq ]\infty, 0^c], \Rightarrow $ either $X$ has minimum in some $n^c$ or $\inf X=\infty$.

$\O\neq X\cap  [\infty, 0^c] \Rightarrow X$ has maximum.

$\O= X\cap  [\infty, 0^c],$ i.e. $ X\subseteq [0, \infty[, \Rightarrow $ either $X$ has maximum in some $n\in \N_0$ or $\sup X=\infty$.

\noindent
Consequently, for a nonempty subset $X\subseteq T$ which has no minimum (respectively, maximum) one has $\inf X=\infty$ (respectively, $\sup X=\infty$).}
\end{remark}

\begin{definition}[Cardinality signature]
For $E\subseteq\N$, the cardinality signature of $E$ is the element $\siga(E)\in T$ defined by
\[
\siga(E)=
\begin{cases}
 k, & \text{if }E\text{ is finite and }|E|=k,\\
 \infty, & \text{if }E\text{ is infinite and }E^c\text{ is infinite},\\
 k^c, & \text{if }E\text{ is cofinite and }|E^c|=k.
\end{cases}
\]\end{definition}

\noindent The signature records the permutation-invariant cardinality information of the set. 

The signature is monotone:
\[
E\subseteq F \quad\Longrightarrow\quad \siga(E)\leq_T \siga(F).
\]
In the present, strict-preferences, case we define the character (see Definition 3.8 of Basile, Bhaskara Rao and Rao \cite{scw}) as follows
\[
\rchi_s:\Pow^s\longrightarrow T,
\qquad
\rchi_s(P)=\siga(D(a,P)).
\]

\begin{theorem}[Canonical  representation]\label{44}
On the strict universal domain $\Pow^s$, 
the class $\SCF_{AN}^s$ consists of
all and only the functions
\[
\phi_C(P)=
\begin{cases}
 a, & \text{if }\rchi_s(P)\in C,\\
 b, & \text{if }\rchi_s(P)\notin C,
\end{cases}
\]
where $C$ is a SOC subset of the linearly ordered set $(T,\leq_T)$.
\end{theorem}
Thus the strict-preference case has a one-dimensional geometrical visualization.  Indeed, {\color{black} since $T$ is linearly ordered, every nonempty proper SOC subset of $T$ is an upper interval, either with a minimum or without a minimum. The only case without a minimum is $]\infty,0^c]$.}

\begin{proof}
See the Appendix. \ref{dim 44}

\end{proof}

 This anticipates what happens with weak preferences. When indifference is allowed, the signature of $D(a,P)$ alone is no longer enough. We must record at the same time the support for $a$ and the support for $b$. Hence the relevant geometry becomes two-dimensional.

\begin{remark}{\rm
\hfill
\begin{enumerate}

\item We observe, to make it unambiguous, that the increase of dimension is due to the possibility of indifference, not to the fact that the voters are countably infinite. 

\item On the other hand, to the infinity of voters
 is due the necessity of replacing  $\{0,1,\dots, n\}$ (relative to the case of $n$ voters) with the set $T$ of more complex, non-trivial, ordered structure. To be explicit, let us recall the  composition 
 of the class $\SCF_{AN}^s$
 in presence of $n$ voters. It is well known that
   $\SCF_{AN}^s$ consists of the {\it quota majority methods} (see Moulin \cite{moulin}, Corollary of page 63). These are the social choice functions $S_k$, with $k\in \{0,1, \dots, n+1\}$,  that are defined by the equivalence \,  $$S_k(P)=a \Leftrightarrow |D(a,P)|\ge k.$$ Mutatis mutandis such functions are nothing else than the 
  $\phi_{\G_0}, \dots, \phi_{\G_k},\dots, \phi_{\G_n},    \phi_{\G_{co-0}}$ corresponding, in the sense of the canonical representation Theorem \ref{44}, to the SOC subsets of  $\{0,1,\dots, n\}$.

\item In the countably infinite case three other kinds of  scfs emerge in $\SCF_{AN}^s$:

-- the function $\phi_{\mathcal G_\infty}$, selecting $a$ if and only if infinitely many voters support $a$;

-- the function $\phi_{\mathcal G_{co-\infty}}$, selecting $a$ if and only if all but finitely many voters support $a$ (equivalently, since preferences are strict in the present section, if and only if finitely many voters support $b$);

-- the function $\phi_{\mathcal G_{co-k}}$ ( $k=1,2,\dots$), selecting $a$ if and only if fewer than $k$ voters support $b$.

Being the remaining scfs \lq\lq same as\rq\rq\, the quota majority rules: $\phi_{\mathcal G_0}$ is the constant rule selecting $a$ for every profile; $\phi_{\mathcal G_k}$ selects $a$ if and only if at least $k$ voters support $a$; $\phi_{\mathcal G_{co-0}}$ is the constant rule selecting $b$ for every profile.

\end{enumerate}}
\end{remark}

\section{The universal weak domain and the infinite grid}\label{sezione5}

We now return to the universal weak domain
$
\Pow=\W(\{a,b\})^\N.
$
Let us denote by $\SCF_{AN}$ the class of all scfs  $\phi:\Pow\to\{a,b\}$ which are anonymous and CSP.

For every profile $P$, the sets $D(a,P)$ and $D(b,P)$ are disjoint, and the remaining voters are indifferent between $a$ and $b$. 
{\color{black}
For the collective-choice analysis below we record the cardinality signatures \[
\siga(D(a,P)),\qquad \siga(D(b,P)).
\]

These do not in general determine the complete orbit of P under permutations of the voters, because the cardinality of the indifferent class is not separately recorded. We shall see, however, that they contain exactly the ordered cardinality information relevant for anonymous CSP rules.
}

Admitting indifference,  the poset relevant for our collective choice analysis is $T\times T$ endowed with the partial ordering defined as follows:

\noindent for $(\alpha,\beta),(\alpha',\beta')\in T\times T$, set
\[
(\alpha,\beta)\sqsubseteq (\alpha',\beta')
\quad\Longleftrightarrow\quad
\alpha\leq_T\alpha'\text{ and }\beta\geq_T\beta'.
\]
Thus moving upward in the order means: more support for $a$, and less support for $b$.
 
\noindent The character (see Definition 3.8 of Basile, Bhaskara Rao and Rao \cite{scw}) is now defined by
\[
\rchi_w:\Pow\longrightarrow T\times T,
\qquad
\rchi_w(P)=\bigl(\siga(D(a,P)),\siga(D(b,P))\bigr),
\]
 clearly the countably infinite counterpart of the map (see \cite{mss}) \(P\mapsto (|D(a,P)|,|D(b,P)|)\)
used in the finite-voter grid representation. 
{\color{black}
Unlike in the finite-voter case, the character $\rchi_w$ does not in general determine the anonymity orbit of a weak-preference profile. Indeed, two profiles may have the same signatures for the coalitions strictly supporting a and b, while the corresponding sets of indifferent voters have different cardinalities. For instance. Let $P$ such that
$D(a, P)$ and $D(b, P)$ form a partition of $\N$ into two infinite sets, so that $I(P)=\O$. Let $Q$ be such that $D(a, Q)$, $D(b, Q)$, and $I(Q)$ are all infinite. Then $\rchi_w(P)=\rchi_w(Q)=(\infty,\infty),$ although $P$ and $Q$ are not related by any permutation of the voters. Nevertheless, as the representation theorem below will show,  the character records exactly the ordered cardinality information relevant for anonymous CSP rules.
}

\begin{proposition}[The range of the weak character]
The range of $\rchi_w$ is  the following set $G$:
\begin{align*}
G={}&(\N_0\times\N_0)
\cup (\N_0\times\{\infty\})
\cup (\{\infty\}\times\N_0)
\cup \{(\infty,\infty)\}\nonumber\\
&\cup \{(k,m^c): k,m\in\N_0,\ k\leq m\} \cup \{(k^c,m): k,m\in\N_0,\ m\leq k\} \nonumber
\end{align*}
\end{proposition}

\begin{proof}
Let $P\in\Pow$. Since $D(a,P)$ and $D(b,P)$ are disjoint, their signatures must be compatible.

If both sets are finite, we get a point of $\N_0\times\N_0$. If one is finite and the other is infinite but not cofinite, we get a point of $\N_0\times\{\infty\}$ or $\{\infty\}\times\N_0$. If both are infinite and not cofinite, we get $(\infty,\infty)$.

If $D(b,P)$ is cofinite and $|D(b,P)^c|=m$, then $D(a,P)$ is finite and contained in $D(b,P)^c$; hence $\siga(D(a,P))=k$ for some $k\leq m$, and the point is $(k,m^c)$. Symmetrically, if $D(a,P)$ is cofinite with $|D(a,P)^c|=k$, then $D(b,P)$ is finite of cardinality $m\leq k$, and the point is $(k^c,m)$.

This proves that $\chi_w(\Pow)\in G$. Conversely, for every point listed in the definition of $G$, it is immediate to construct two disjoint subsets $E,F\subseteq\N$ having the required signatures. Then define a profile $P$ by making the voters in $E$ strictly prefer $a$, the voters in $F$ strictly prefer $b$, and all remaining voters indifferent. Hence every point of $G$ belongs to the range of $\rchi_w$.
\end{proof}
\noindent
The range of $\rchi_w$ can be visualized as the infinite triangular grid of Figure 2.

$$
\begin{tikzpicture}[xscale=0.35, yscale=0.28]
\draw [help lines][dashed, ultra thin] (0,0) grid (40,40);
\draw [<->] (45,0) -- (0,0) -- (0,45);
\draw (40,0) -- (0,40);
\draw [thick](20,0) -- (20,20) -- (0,20);
\draw [fill=white, ultra thick,white ] (40,40) -- (40,0) -- (0,40) -- (40,40);
\node at (20,-1) {$\infty$}; \node at (-1,20) {$\infty$}; \node at (10,-1){...}; \node at (18.2,-1){...}; \node at (30,-1){...}; \node at (22.3,-1){...};
\draw  (40,0) -- (0,40);
{\footnotesize 
\node at (0,-1) {$0$};
\node at (1,-1) {$1$};
\node at (2,-1) {$2$};\node at (3,-1) {$3$}; \node at (4,-1) {$4$};
\node at (40,-1) {$0^c$};
\node at (39,-1) {$1^c$};\node at (38,-1) {$2^c$}; \node at (37,-1) {$3^c$};\node at (36,-1) {$4^c$};
\node at (-1,0) {$0$}; \node at (-1,1) {$1$}; \node at (-1,2) {$2$}; \node at (-1,3) {$3$};\node at (-1,4) {$4$};
\node at (-1,40) {$0^c$};\node at (-1,39) {$1^c$}; \node at (-1,38) {$2^c$};\node at (-1,37) {$3^c$};\node at (-1,36) {$4^c$}; \node at (-1,10) {$\vdots$}; \node at (-1,18) {$\vdots$}; \node at (-1,22) {$\vdots$}; \node at (-1,30) {$\vdots$};
\node at (-1,23.5) {$n^c$};
\node at (16.5,-1) {$n$}; \node at (25,-1) {$k^c$}; \node at (-1,15) {$k$};
\draw (0,23.5)-- (16.5,23.5); \draw (16.5,0)-- (16.5,23.5);
\draw (0,15)-- (25,15); \draw (25,15)-- (25,0);
}
\node at (20,-3) {Figure 2: the grid $G=\rchi_w(\Pow)$};
\node at (44,-2) {$T$}; \node at (-2,44) {$T$};
\end{tikzpicture}
$$
The restriction of  $\sqsubseteq$  to $G$ is assumed to be the order $\leq_G$, and the canonical scf  (see Definition 3.9 of Basile, Bhaskara Rao and Rao \cite{scw}) is defined next.
\begin{definition}
For every  subset $C$ of the poset $( G, \le_G)$,  define the  scf $\phi_C$ by means of 
\[
\phi_C(P)=a
\quad\Longleftrightarrow\quad
\rchi_w(P)\in C.
\]
\end{definition}
\noindent Observe that the rule is anonymous because the signature of a subset of $\N$ is invariant under permutations of the voters. Conversely, every social choice function constant on each fiber of $\rchi_w$ is of this form for a unique subset $C\subseteq G$.

\noindent We can now state the main representation theorem.

\begin{theorem}[Infinite-grid canonical representation]\label{53}
Let $\phi:\Pow\to\{a,b\}$ be a scf. Then $\phi\in\SCF_{AN}$ if and only if there exists a unique SOC subset $C$ of the poset $(G,\leq_G)$ such that
$
\phi=\phi_C$.
\end{theorem}
Its proof is based on the results of Section 3 of Basile, Bhaskara Rao and Rao \cite{scw} and  on the following lemma which is the set-theoretic fact that makes the order on the grid compatible with anonymity. It is this technical point that allows us to pass from comparisons of signatures to comparisons of actual coalitions, after a suitable relabelling of the voters. The proof of the lemma is postponed to the Appendix. \ref{dim permutation lemma}

\begin{lemma}[permutation lemma]\label{permutation lemma}
Let $A_1,B_1,A_2,B_2\subseteq\N$ be such that
\[
A_1\cap B_1=\varnothing,
\qquad
A_2\cap B_2=\varnothing.
\]
Assume
\[
\siga(A_1)\leq_T\siga(A_2),
\qquad
\siga(B_1)\geq_T\siga(B_2).
\]
Then there exists a permutation $\pi$ of $\N$ such that
\[
\pi(A_1)\subseteq A_2,
\qquad
B_2\subseteq \pi(B_1).
\]
\end{lemma}

\begin{proof} [Proof of Theorem \ref{53}]
\noindent
First let $C$ be a SOC subset of $(G,\leq_G)$. We prove that $\phi_C$ is CSP by appealing to Theorem 3.3 in Basile, Bhaskara Rao and Rao \cite{scw}, namely by showing the almost monotonicity of $\phi_C$. 

First observe that, due to Proposition 3.11 in Basile, Bhaskara Rao and Rao \cite{scw}, $\phi_C$ is $(\rchi_w,a)$-monotone (Definition 3.10 in Basile, Bhaskara Rao and Rao \cite{scw}). We recall that this means that if for the profiles $P, Q$ one has $\rchi_w(P)\ge_G \rchi_w(Q)$ and $\phi_C(Q)=a$, then $\phi_C(P)=a$.

\noindent Now, to show the almost monotonicity of $\phi_C$, let us assume that $\quad  \left\{ 
\begin{array} {ll}
D(a, Q)\subseteq D(a, P)  \\

D(b, Q)\supseteq D(b, P) .\\
\end{array}
\right.
$ 
Then by monotonicity of the signature one has $\rchi_w(P)\ge_G \rchi_w(Q)$. So the almost monotonicity of  $\phi_C$ simply comes from its $(\rchi_w,a)$-monotonicity.

\bigskip
\noindent
Conversely, assume that $\phi\in\SCF_{AN}$. If we show that $\phi$ is $(\rchi_w,a)$-monotone,  then by applying again Proposition 3.11 in Basile, Bhaskara Rao and Rao \cite{scw}, we obtain the assertion.

\noindent So, let us assume that for the profiles $P, Q$ one has $\rchi_w(P)\ge_G \rchi_w(Q)$ and $\phi(Q)=a$. We have to show that $\phi(P)=a$.

The inequality $\rchi_w(P)\ge_G \rchi_w(Q)$ allows to apply Lemma \ref{permutation lemma} determining the existence of a permutation $\pi$ of $\mathbb N$ such that 
$$\quad  \left\{ 
\begin{array} {ll}
D(a, Q)\subseteq D(a, P\circ \pi)  \\

D(b, Q)\supseteq D(b, P\circ \pi) .\\
\end{array}
\right.
$$
Since $\phi$ is CSP, again by Theorem 3.3 in Basile, Bhaskara Rao and Rao \cite{scw} we get 
$\phi(P\circ \pi)=a$, and since $\phi$ is also anonymous, we also get $\phi(P)=a.$
\end{proof}
{\color{black}A direct consequence of Theorem \ref{53} is that, for any fixed anonymous CSP rule, all profiles at which both strict-support coalitions are infinite receive the same outcome, since their character is $(\infty,\infty)$. This remains true whether the indifferent class is empty, finite, or infinite. The conclusion uses both axioms, since profiles with different cardinalities of the indifferent class need not be related by a permutation. The outcome at this common character may be either $a$ or $b$, depending on the rule. Anonymity imposes symmetry among voters; the resolution of this common character remains part of the specification of the rule.
}

The representation theorem reduces the description of anonymous CSP binary social choice functions with countably many voters to the description of the SOC subsets of the infinite grid $G$. The latter is a geometric and order-theoretic problem. In the finite-voter case, this is where the vertical-horizontal zigzag representation of Basile, Rao and Rao \cite{mss} enters. In the present countably infinite setting, the geometric description is richer because the grid contains finite, infinite non-cofinite, and cofinite signatures. It will be developed separately, completing the extension of \cite{mss}.

 \section{The geometry of the SOC subsets of  $(G,\leq_G)$}\label{sei}
 
 For a social choice function $\phi=\phi_C$, with $C\subseteq G$, the set
 $G(a)=\{\chi_w(P)\in G:\phi(P)=a\}=C$
is the region of the grid where the social choice is \(a\). The representation theorem \ref{53} says that \(\phi\) is coalition strategy-proof exactly when \(G(a)\) is a SOC subset of $(G,\leq_G)$. We therefore turn to the geometric description of the SOC subsets of this grid.

The following analysis reveals that such a SOC region must have one of the generalized zigzag shapes I--IV (Section \ref{four geometries}, {\it infra}). Thus, the finite-voter zigzag picture of \cite{mss} is preserved, but it acquires two possible missing-boundary components, the \(J_{k,\infty}\)- and \(K_{\infty,\ell}\)\,-regions defined next, which are specific to the countably infinite society.
{\color{black}
A first essential difference from the finite case is that a SOC subset of $G$ need not be generated by the antichain of its minimal elements; consequently, the geometric description requires additional boundary data.
}

\subsection{Some elementary regions of $G$ and corners}
Set 
\[
E=\{(\alpha,\beta)\in G:\alpha<\infty,
\ \beta>\infty\},
\] and \[
F=\{(\alpha,\beta)\in G:\alpha>\infty,
\ \beta<\infty\}.
\]

With reference to figure 2, these are, respectively,  the north-west triangle of vertices $(0,\infty), (0, 0^c), (\infty, \infty)$ without the whole base, and the south-east triangle of vertices $(\infty,0),  (\infty, \infty) (0^c, 0)$ without the whole height.

By setting $
\Nb=\Nzero\cup\{\infty\},
$
we also explicitly write the decomposition  

\[
G=E\cup\Nb^2\cup F.
\]

 \begin{definition}\label{61}
For \((\alpha,\beta)\in G\), set
\[
I_{\alpha,\beta}
=\up(\alpha,\beta)
=\{(\alpha',\beta')\in G:\alpha'\ge_T\alpha,
\ \beta'\le_T\beta\}.
\]
For \(k\in\Nb\), set
\[
J_{k,\infty}
=\{(\alpha,\beta)\in G:\alpha\ge_T k,
\ \beta<_T\infty\}.
\]
For \(\ell\in\Nb\), set
\[
K_{\infty,\ell}
=\{(\alpha,\beta)\in G:\alpha>_T\infty,
\ \beta\le_T\ell\}.
\]
\end{definition}
\noindent Thus \(J_{k,\infty}\) is obtained from \(I_{k,\infty}\) by deleting the whole base line at ordinate \(\infty\).
The region \(K_{\infty,\ell}\) is obtained from \(I_{\infty,\ell}\) by deleting the whole height line at abscissa \(\infty\).
Notice that
\[
K_{\infty,\ell}\subseteq J_{k,\infty} \mbox{ for every }k, \ell\in \Nb, \mbox{ and }
K_{\infty,\infty}=F.
\]


{\color{black}
The regions just defined have a direct interpretation in terms of strict support, strict opposition, and indifference. For $k\in\mathbb N_0$,

$$
\chi_w(P)\in J_{k,\infty}
\quad\Longleftrightarrow\quad
|D(a,P)|\ge k
\quad\text{and}\quad
|D(b,P)|<\infty.
$$

Thus $J_{k,\infty}$ combines a threshold on strict support for $a$ with the requirement that only finitely many voters strictly oppose $a$. When $k=\infty$, the support condition requires infinitely many strict supporters of $a$. In either case, the indifferent class may be finite or infinite.

For $\ell\in\mathbb N_0$,

$$
\chi_w(P)\in K_{\infty,\ell}
\quad\Longleftrightarrow\quad
|D(b,P)|\le\ell
\quad\text{and}\quad
|D(b,P)\cup I(P)|<\infty.
$$

The second condition says that all but finitely many voters strictly support $a$: both strict opponents and indifferent voters belong to the finite exceptional set. The first condition additionally bounds the number of strict opponents. When $\ell=\infty$, only the cofiniteness requirement remains, giving $K_{\infty,\infty}=F$.

The distinction between finite opposition and cofinite strict support is possible because indifference is admitted. Without indifferent voters, the two properties coincide. For example, a profile with infinitely many strict supporters of $a$, finitely many strict supporters of $b$, and infinitely many indifferent voters belongs to $J_{\infty,\infty}$ but to no $K_{\infty,\ell}$.

These interpretations also clarify the missing boundaries. Passing from $I_{k,\infty}$ to $J_{k,\infty}$ excludes the boundary profiles where the strict opposition has signature $\infty$. Passing from $I_{\infty,\ell}$ to $K_{\infty,\ell}$ excludes those where strict support for $a$ has signature $\infty$. The resulting requirements impose finiteness without a common finite bound: on the number of strict opponents in the first case, and on the number of voters who fail to strictly support $a$ in the second.

For the elementary rules $\phi_{J_{k,\infty}}$ and $\phi_{K_{\infty,\ell}}$, the respective membership conditions are necessary and sufficient for choosing $a$. When these regions occur as components of a union representing a SOC set $C$, each supplies a sufficient condition for $\phi_C$ to choose $a$.
}

{\color{black}
\noindent 
For a numerical illustration, consider the rule whose winning region is
\[
S=\uparrow\{(3,3),(9,7),(11,13)\}\cup J_{16,\infty},
\]
depicted later in Figure 4. In this and later similar  descriptions, supporters (of $a$) strictly prefer $a$ to $b$, opponents strictly prefer $b$ to $a$, and all remaining voters are indifferent.

This rule chooses \(a\) if and only if
at least one of four conditions holds: there are at least 3 supporters and at most 3 opponents; at least 9 supporters and at most 7 opponents; at least 11 supporters and at most 13 opponents; or at least 16 supporters and only finitely many opponents. The first three conditions define the region generated by the three displayed grid points. The fourth is the contribution of $J_{16,\infty}$: once support reaches 16 voters, opposition may exceed any prescribed finite bound, provided that it remains finite.

For example, with 12 supporters and 14 opponents the outcome is $b$, whereas with 16 supporters and 14 opponents it is $a$. With 16 supporters and infinitely many opponents, the outcome is again $b$, illustrating the excluded boundary. The rule thus allows progressively greater opposition as support increases. Each condition is preserved when the support coalition expands and the opposition coalition contracts; the resulting winning region is SOC and therefore defines an anonymous coalition strategy-proof rule by Theorem 5.3.
}

For a subset \(S\subseteq G\), we use the standard notation
$\,
S_\alpha=\{\beta\in T:(\alpha,\beta)\in S\},
$\, and
$\,
S^\beta=\{\alpha\in T:(\alpha,\beta)\in S\}
$\,
for, respectively, the vertical section at abscissa \(\alpha\), and
the horizontal section at ordinate \(\beta\).

As recalled in \cite{scw}, SOC subsets of a finite poset are generated by the antichain made of their minimal elements. In the present infinite grid this is no longer true. This is the reason why to determine the shape of the SOC subsets of $G$ we cannot refer only to minimal points.

\begin{example}[{\it A SOC subset not generated by an antichain}]\label{62}
{\rm Consider the following rule: the collective choice is $a$ if and only if  the voters supporting $b$ or  indifferent between the two alternatives are finitely many. In other words: the set of voters in favor of $a$ is cofinite.
The corresponding set $G(a)$ is the set $F$ which is SOC. 
However, it cannot be generated by an antichain.}

\noindent {\rm To see this, let $A\subseteq F$ be an antichain. Because of next Proposition \ref{prop:infinite-antichains} $A$ is necessarily finite. But in this case $A$ cannot generate all points of $F$.
}
\end{example}
\noindent On the other hand, $F$ has a very simple generator: 
\[
F=\up \{(k^c,k):k\in\Nzero\} .
\]
Indeed, if $(k^c,m)\in F$, then $m\le k$, and therefore
\[
(k^c,m)\ge_G(k^c,k).
\]
Thus $F$ is generated by the infinite diagonal segment $\{(k^c,k):k\in\Nzero\}$.\qed 

\bigskip
\begin{proposition}[Infinite antichains of \(G\)]\label{prop:infinite-antichains}

The infinite antichains of $G$ are all and only  the subsets $A$ of $G$ of the following form: 

-- $A=\{(k_1, m_1), (k_2, m_2), (k_3, m_3),\dots \}$ where the $k$'s and the $m$'s are strictly increasing sequences in $\N_0$

-- $A=\{(k_1, m_1), (k_2, m_2), (k_3, m_3),\dots\} \cup \{(\infty, \infty)\} $ where the $k$'s and the $m$'s are strictly increasing sequences in $\N_0$
\end{proposition}
\begin{proof} See the Appendix. \ref{dim prop:infinite-antichains}\end{proof}

While it is natural to refer as \lq\lq corners\rq\rq\,  to minimal points of a SOC subset $S$ of $G$, Example \ref{62} forces the necessity to consider other kinds of corners to describe the shape of a general SOC subset of $G$.

\noindent
Let \(S\) be a nonempty SOC subset of \(G\). A point \((k,\infty)\in G\), with \(k\in\Nb\), that has the two properties
\[
J_{k,\infty}\subseteq S,
\qquad
(k,\infty)\notin S,
\]
is a missing-boundary corner: the region below the missing point is present, but the boundary point itself is not. Analogously, for a point \((\infty,\ell)\in G\), with \(\ell\in\Nb\),  if
\[
K_{\infty,\ell}\subseteq S,
\qquad
(\infty,\ell)\notin S,
\]
we have a symmetric missing-boundary case.
\begin{remark} {\rm If we define
\[
A_2(S)=\{k\in\Nb:J_{k,\infty}\subseteq S,
\ (k,\infty)\notin S\},
\]
 then \(A_2(S)\) has a minimum if it is nonempty.
 Similarly, if we define
\[
A_3(S)=\{\ell\in\Nb:K_{\infty,\ell}\subseteq S,
\ (\infty,\ell)\notin S\},
\]
 then \(A_3(S)\) has a maximum if it is nonempty.\footnote{ The existence of this maximum follows from the fact that
\[
K_{\infty,\infty}=\bigcup_{\ell\in\Nzero}K_{\infty,\ell}.
\]}
 }\end{remark}

\noindent
Let us use the notations $\tau_1(S)$ for the set of minimal elements of $S$, and

\[
k_2(S)=\min A_2(S),\quad
\ell_3(S)=\max A_3(S).
\]
Keeping in mind that $S$ may have or not minimal elements, as well as the sets \(A_i(S)\) may be empty or not, we introduce the corner types definition.

\bigskip
\begin{definition}[Corners]
Let \(S\) be a nonempty SOC subset of \(G\).
The set $\tau(S)$ of the corners of  \(S\)  consists of the minimal elements of $S$, $(k_2(S),\infty)$, and $(\infty,\ell_3(S))$. We distinguish corners as of type 1 in case of minimal elements, of type 2 for $(k_2(S),\infty)$, of type 3 for $(\infty,\ell_3(S))$.

\end{definition}
\noindent
In formula
\[
\tau(S)
=
\tau_1(S)
\cup
\{(k_2(S),\infty):A_2(S)\ne\varnothing\}
\cup
\{(\infty,\ell_3(S)):A_3(S)\ne\varnothing\}.
\]
Each of the three components of $\tau(S)$ can be empty, of course. 
\begin{definition}
Let \(S\) be a nonempty SOC subset of \(G\). We say that a point of $S$ is absorbed by a corner of $S$ when it is contained either in $I_{\alpha,\beta}$, where \((\alpha,\beta)\in \tau_1(S)$, or in
\(J_{k_2(S),\infty}\),  or in
\(K_{\infty,\ell_3(S)}\).
\end{definition}

\noindent
The following is the key for the geometric description of the SOC subsets of $G$,  set forth in Theorem \ref{prop:exact} and Corollary \ref{thm:zigzag}.
\begin{lemma}[Covering lemma]\label{lem:covering}
Let \(S\) be a nonempty SOC subset of \(G\). Then every point of \(S\) is absorbed by some corner.
\end{lemma}

\begin{proof} See the Appendix. \ref{dim lem:covering}\end{proof}


\begin{theorem}[Exact representation by corners]\label{prop:exact}
Let \(S\) be a nonempty SOC subset of \(G\). Then
\[
S=
\up\tau_1(S)
\cup
\begin{cases}
J_{k_2(S),\infty},& if\, A_2(S)\ne\varnothing,\\
\varnothing,& if\,  A_2(S)=\varnothing,
\end{cases}
\cup
\begin{cases}
K_{\infty,\ell_3(S)},& if\, A_3(S)\ne\varnothing,\\
\varnothing,& if\, A_3(S)=\varnothing.
\end{cases}
\]
\end{theorem}
\begin{proof}
The inclusion ``\(\supseteq\)'' follows directly from the definitions: type 1 corners belong to \(S\), and their upward closures are contained in \(S\); the  \(J\)- and \(K\)-regions, when present, are contained in \(S\) by definition. The reverse inclusion is nothing else than  Lemma~\ref{lem:covering}.
\end{proof}

\begin{remark}\label{interpretazione1}
 {\rm
We emphasize the interpretation of the above formula. The set \(S\) is the union
of the SOC set generated by the minimal elements of \(S\), if there are any,
together with the \(J\)-region \(J_{k_2(S),\infty}\), if the corner of type 2 exists,
and the \(K\)-region \(K_{\infty,\ell_3(S)}\), if the corner of type 3 exists.
Each of these three components is omitted when the corresponding kind of
corner does not exist.
The two exceptional regions may both arise from the definition of corners of
type 2 and type 3.\footnote{For example, for $S=J_{0, \infty}$ the point $(0,\infty)$ is the  corner of type 2 and the point $(\infty, \infty)$ is the  corner of type 3.}
 However, they need not be retained simultaneously in a nonredundant representation, since every $J$-region contains every $K$-region. Moreover, if a minimal element of \(S\) lies in the North-West
triangle \(E\), then the \(K\)-term is already redundant, even without a
\(J\)-term, because such a minimal element generates the whole relevant
South-East cofinite-finite part. After retaining at most one exceptional region, that region is itself omitted whenever it is contained in $\uparrow\tau_1(S)$. Any exceptional region that remains therefore contributes essentially to the representation.
Theorem \ref{prop:exact} thus identifies precisely how the finite-grid description must be extended. Although a SOC subset of $G$ need not be generated by its minimal elements, every point not generated by them is accounted for by the specified missing-boundary regions.  Together with Theorem \ref{53}, this gives an exhaustive description of the anonymous coalition strategy-proof rules in terms of genuine corners and, when necessary, one essential exceptional region.
}\end{remark}

{\color{black}
In the foregoing, we have started from a SOC subset $S$ of $G$, identified its genuine and missing-boundary corners, and explained how to remove redundant terms from its representation. We now formulate the construction in the opposite direction. We specify an antichain and an optional exceptional region, subject to two admissibility conditions, and use them to generate a SOC subset of $G$. These conditions ensure that the chosen antichain is precisely the set of genuine corners and that the exceptional region, when present, contributes essentially. The resulting correspondence is one-to-one.

\begin{definition}[Admissible parameters]\label{definition of parameters}
Let
$$
\mathcal R=
\{\varnothing\}
\cup
\{J_{k,\infty}:k\in\widehat{\mathbb N}_0\}
\cup
\{K_{\infty,\ell}:\ell\in\widehat{\mathbb N}_0\}.
$$
A pair \((A,R)\), where \(A\subseteq G\) is an antichain and \(R\in\mathcal R\), is called admissible if
$$
A\cap R=\varnothing, \mbox{ \, and either }
R=\varnothing
\quad\text{ or }\quad
R\not\subseteq\uparrow A.
$$
\end{definition}

\noindent
In the above definition the antichain \(A\) is allowed to be empty, with the convention
\(\uparrow\varnothing=\varnothing\). 
Denote by \(\mathcal D\) the family of admissible pairs. For every
\((A,R)\in\mathcal D\), define

$$
S_{A,R}=\big(\uparrow A\big)\cup R.
$$

\begin{proposition}[Nonredundant parametrization]\label{use of parameters}

 The assignment

$$
(A,R)\longmapsto S_{A,R}
$$
is a bijection between \(\mathcal D\) and the family of all SOC subsets of \(G\). Moreover, \(A\) is exactly the set of minimal elements of \(S_{A,R}\). The empty SOC set corresponds to the pair
\((\varnothing,\varnothing)\).
\end{proposition}

\begin{proof}
See the Appendix \ref{dim use of parameters}
\end{proof}
}
{\color{black}
At this point, we can express the anonymous coalition strategy-proof rules directly in a more explicit parametric form. The validity of the parametrization requires no new proof, as it is merely a restatement of Definition \ref{definition of parameters} and Proposition \ref{use of parameters}, in the light of Theorem \ref{53}.
\noindent
First, let us introduce the chain ${\widehat T}$
\[
\star<_{\widehat T} 0<_{\widehat T} 1<_{\widehat T} 2<_{\widehat T}\cdots
<_{\widehat T}\infty_K<_{\widehat T}\infty_J
<_{\widehat T}\cdots<_{\widehat T} 2^c<_{\widehat T} 1^c<_{\widehat T} 0^c
\]
obtained from $T$ by splitting $\infty$ into two adjacent
symbols $\infty_K,\infty_J$ and adjoining the symbol $\star$ as a new minimum.
We use its elements to index the missing-boundary regions by defining, 
for $r\in{\widehat T}$,
\[
R_r=
\begin{cases}
\varnothing, & r=\star,\\
K_{\infty,\ell}, & r=\ell\in\mathbb N_0,\\
K_{\infty,\infty}, & r=\infty_K,\\
J_{\infty,\infty}, & r=\infty_J,\\
J_{k,\infty}, & r=k^c,\quad k\in\mathbb N_0.
\end{cases}
\]
The order on ${\widehat T}$ records inclusion of the exceptional regions:
\[
r<_{\widehat T} r'
\quad\Longleftrightarrow\quad R_r\subset R_{r'}.
\]
An anonymous CSP scf can be specified by a parameter configuration $$\Gamma  = \langle \{(\alpha_1, \beta_1), (\alpha_2, \beta_2), \dots\}; r\rangle$$
 provided that  $(\alpha_1, \beta_1),  (\alpha_2, \beta_2), \dots$ form an antichain in $G$ (empty, finite or infinite) and, for $r\in\widehat T$, either $r=\star$ or the following two conditions hold: 
 
 $
R_r\cap \{(\alpha_1, \beta_1), (\alpha_2, \beta_2), \dots\}=\varnothing,
\qquad$ 
 and $\quad R_r\not\subseteq\uparrow \{(\alpha_1, \beta_1), (\alpha_2, \beta_2), \dots\}.
$

\noindent
The resulting one-to-one parametrization  is
\qquad
$\Gamma\mapsto\phi_{\Gamma}$ \qquad where $\phi_{\Gamma}$ is the canonical scf associated to the SOC subset of $G$ defined as
$R_r\cup\uparrow \{(\alpha_1, \beta_1), (\alpha_2, \beta_2), \dots\}.
$
{\color{black}
The parameters have a direct threshold interpretation. If \(\chi_w(P)=(\alpha,\beta)\), each genuine corner \((\alpha_i,\beta_i)\) specifies a lower threshold \(\alpha_i\) for the strict-support signature and an upper threshold \(\beta_i\) for the strict-opposition signature. The corresponding condition is

$$
\alpha\ge_T\alpha_i
\qquad\text{and}\qquad
\beta\le_T\beta_i.
$$

When both thresholds are finite, this means that at least \(\alpha_i\) voters strictly support \(a\) and at most \(\beta_i\) strictly oppose it. The rule chooses \(a\) exactly when at least one of these threshold pairs is met or \(\chi_w(P)\in R_r\).

The index \(r\) specifies the optional additional condition. For \(r=\star\), the genuine-corner conditions alone determine the choice of \(a\). A finite index \(r=\ell\) adds the sufficient condition that all but finitely many voters strictly support \(a\) and at most \(\ell\) strictly oppose it. At \(r=\infty_K\), support remains cofinite, while finite opposition is allowed without a common finite ceiling. At \(r=\infty_J\), the additional condition requires infinitely many strict supporters and only finitely many strict opponents; the indifferent class may now be infinite. Finally, \(r=k^c\) specifies a minimum of \(k\) strict supporters together with finite strict opposition.

The order on \(\widehat T\) consequently has a choice interpretation. Holding the genuine-corner antichain fixed, increasing \(r\) between two admissible configurations strictly enlarges the set of profiles at which \(a\) is chosen. Along the finite branch, this relaxes the ceiling on opposition while retaining cofinite support; the passage from \(\infty_K\) to \(\infty_J\) permits infinitely many indifferent voters; and along the branch indexed by \(k^c\), it lowers the minimum support requirement while retaining finite opposition. Inclusion follows from the nesting of the regions \(R_r\), and strictness follows from the nonredundancy of the parametrization.

}
}

\subsection{The four geometric families and prototypical figures}\label{four geometries}

We now describe the possible shapes of a nonempty SOC subset \(S\subseteq G\) in terms of its genuine corners \(\tau_1(S)\) and the two possible missing-boundary regions.
The empty SOC set is to be added separately; in the social choice interpretation it corresponds to the constant rule choosing \(b\).

\begin{corollary}[Generalized zigzag description of SOC subsets of \(G\)]\label{thm:zigzag}
A nonempty SOC subset $S$ of \(G\) has one of the following four forms.

\begin{enumerate}
\item[\emph{I.}] 
If the set \(\tau_1(S)\) is infinite, then 
$S=\up \tau_1(S).
$
\item[\emph{II.}] 
If the set \(\tau_1(S)\) is finite and
$\,
\tau_1(S)\subseteq\Nb^2,$\,
then \,
$
S=\up \tau_1(S),
\quad
S=\up \tau_1(S)\cup J_{k,\infty},
\quad
\hbox{or}
\quad
S=\up \tau_1(S)\cup K_{\infty,\ell},
$
for suitable \(k,\ell\in\Nb\).

\item[\emph{III.}] If the set \(\tau_1(S)\) is finite and
$
\tau_1(S)\cap E\ne\varnothing,
$
then
$
S=\up \tau_1(S)
\quad
\hbox{or}
\quad
S=\up \tau_1(S)\cup J_{k,\infty}
$
for some \(k\in\Nb\). 

\item[\emph{IV.}] If the set \(\tau_1(S)\) is finite and
$
\tau_1(S)\cap F\ne\varnothing,
$
then
$
S=\up \tau_1(S)
\quad
\hbox{or}
\quad
S=\up \tau_1(S)\cup K_{\infty,\ell}
$
for some \(\ell\in\Nb\). 
\end{enumerate}
\end{corollary}
\begin{proof} See the Appendix. \ref{dim thm:zigzag} \end{proof}
Thus, if the set \(\tau_1(S)\) is infinite, we are dealing with a genuine-corners antichain that can only be
$\{(k_i,m_i):i\in\N\}$ or $\{(k_i,m_i):i\in\N\}\cup\{(\infty,\infty)\},
$
where \(k_i,m_i\in\Nzero\) and both sequences are strictly increasing. Geometrically, $S$ 
looks like in Figure 3. Namely as a \lq\lq triangle\rq\rq\, that may include or not the vertex $(\infty,\infty)$ depending on the fact that $(\infty,\infty) \in$ or $\notin \tau_1(S)$.  The other two vertices of the \lq\lq triangle\rq\rq\,  are  $(k_1,0)$ and $(0^c, 0)$. The left side of the triangle is a vertical-horizontal zigzag line with the corners $(k_i, m_i)$. In case there are finitely many, even zero, minimal points this gives the remaining cases. When the genuine-corner antichain is contained in the central square,  if both $A_2(S)$ and $A_3(S)$ are nonempty, since every $K$-region is contained in every $J$-region, the \(K_{\infty, \ell_3(S)}\)-term  may be omitted. Figures 4 and 5 illustrate the cases. Figure 6 illustrates when  the finite genuine-corner antichain meets the north-west triangle. In this case a possible \(K\)-term is redundant, because any point of \(\tau_1(S)\cap E\) already generates every region \(K_{\infty,\ell}\). Finally the case IV corresponds to Figure 7: at least a minimal point of $S$ lies in south-east triangle, what determines that no  \(J\)-term can occur: the presence of a \(J\)-region would destroy the minimality  of any point of \(\tau_1(S)\cap F\).


$$\begin{tikzpicture}[xscale=0.35, yscale=0.28]
\draw [<->] (45,0) -- (0,0) -- (0,45);
\draw [dashed, ultra thin](19.8,20)  -- (0,20);
\draw [fill=white, ultra thick,white ] (40,40) -- (40,0) -- (0,40) -- (40,40);
\node at (20,-1) {$\infty$}; \node at (-1,20) {$\infty$}; \node at (10,-1){...}; \node at (18.2,-1){...}; \node at (30,-1){...}; \node at (22.3,-1){...};
\draw  (40,0) -- (0,40);
\fill[gray!20] (1,0)-- (1,3)--(3,3)--(3,5)--(8,5)--(8,9)--(10,9)--(10,10)--(14,10)--(14,14)--(15,14)--(15,17)--(16,17)--(16,18)--(17,18)--(17,19)--(18,19)--(18, 19.5)
  -- (20,20) -- (40,0) -- cycle;
\draw  (40,0) -- (0,40);
{\footnotesize 
\node at (0,-1) {$0$};
\node at (1,-1) {$1$};
\node at (2,-1) {$2$};\node at (3,-1) {$3$}; \node at (4,-1) {$4$};
\node at (40,-1) {$0^c$};
\node at (39,-1) {$1^c$};\node at (38,-1) {$2^c$}; \node at (37,-1) {$3^c$};\node at (36,-1) {$4^c$};
\node at (-1,0) {$0$}; \node at (-1,1) {$1$}; \node at (-1,2) {$2$}; \node at (-1,3) {$3$};\node at (-1,4) {$4$};
\node at (-1,40) {$0^c$};\node at (-1,39) {$1^c$}; \node at (-1,38) {$2^c$};\node at (-1,37) {$3^c$};\node at (-1,36) {$4^c$}; \node at (-1,10) {$\vdots$}; \node at (-1,18) {$\vdots$}; \node at (-1,22) {$\vdots$}; \node at (-1,30) {$\vdots$};
\draw (1,0)-- (1,3)--(3,3)--(3,5)--(8,5)--(8,9)--(10,9)--(10,10)--(14,10)--(14,14)--(15,14)--(15,17)--(16,17)--(16,18)--(17,18)--(17,19)--(18,19)--(18, 19.5); 
\draw [->](18.5,18.5)-- (19.5,19.5);
\node at (1,3){$\star$};\node at (3,5){$\star$};\node at (8,9){$\star$};\node at (1,3){$\star$};\node at (10,10){$\star$};\node at (14,14){$\star$};\node at (15,17){$\star$};\node at (16,18){$\star$};\node at (17,19){$\star$};
\node at (22.3,20) {$(\infty, \infty)$};
\draw[thick] (20,20) circle [radius=0.27];
\draw [dashed, ultra thin](20,0) -- (20,19.8);

}
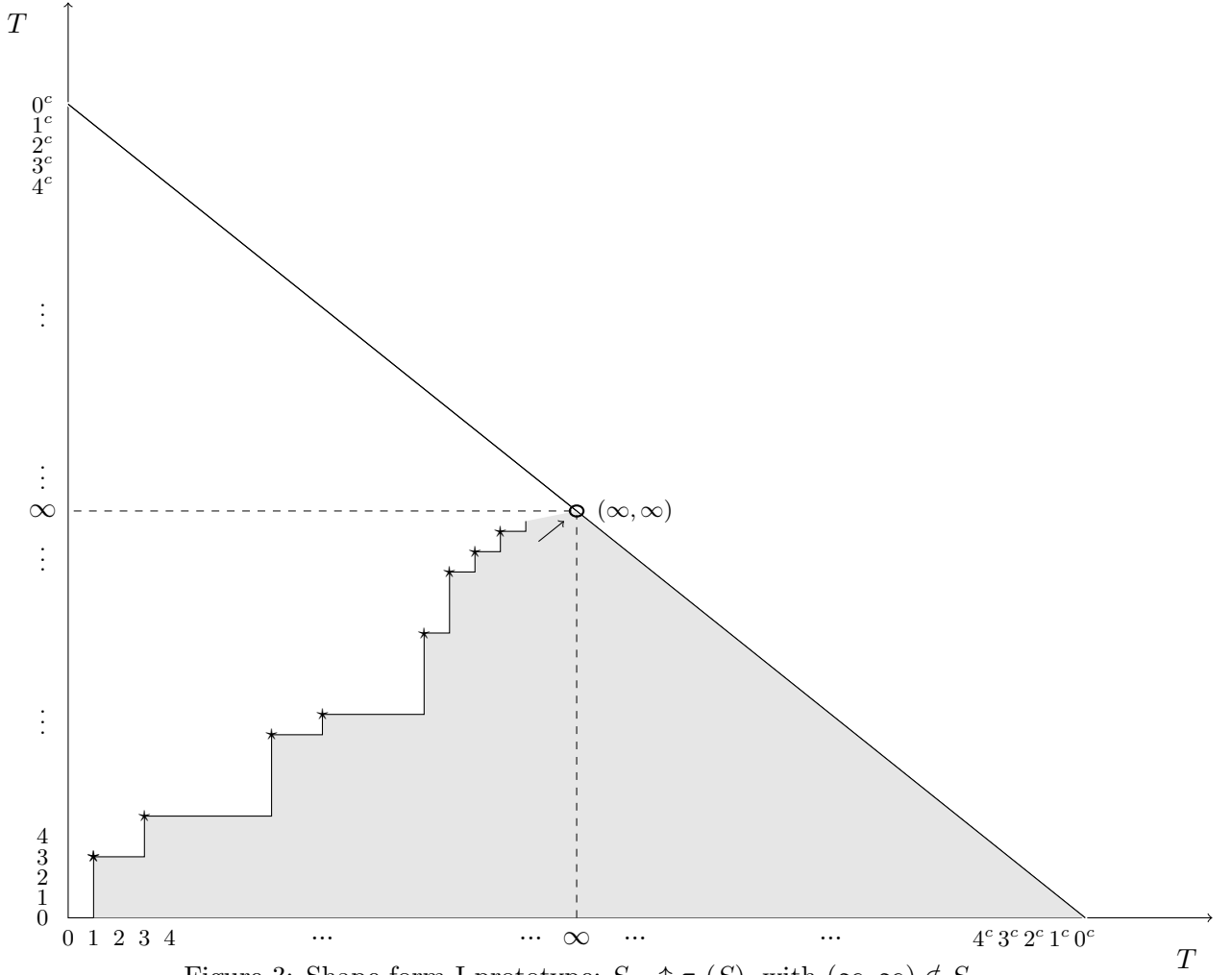
\node at (20,-3) {Figure 3: Shape form I prototype;  $S=\uparrow \tau_1(S)$, with $(\infty,\infty)\notin S $};

\node at (44,-2) {$T$}; \node at (-2,44) {$T$};
\end{tikzpicture}
$$
{\color{black}
The alternative case $(\infty,\infty)\in S $  is obtained by including $(\infty,\infty)$ as an additional genuine corner, and all other grid-point memberships remain unchanged. 
Writing the genuine corners of the depicted set as $(k_i,m_i), i\in\mathbb N$, the rule chooses $a$ if and only if some threshold pair is met: at least $k_i$ voters strictly support $a$ and at most $m_i$ strictly oppose it. Including the vertex $(\infty,\infty)$ additionally selects $a$ at every profile where infinitely many voters strictly prefer each alternative, regardless of how many voters are indifferent. All other outcomes remain unchanged.}


\begin{center}
\begin{tikzpicture}[xscale=0.35, yscale=0.28, line cap=round, line join=round]


\draw[<->] (45,0) -- (0,0) -- (0,45);
\draw (40,0) -- (0,40);

\fill[white] (40,40) -- (40,0) -- (0,40) -- cycle;
\draw (40,0) -- (0,40);

\fill[gray!20] (3,0) -- (3,3) -- (9,3) -- (9,7) -- (11,7) -- (11,13)
  -- (27,13) -- (40,0) -- cycle;
\fill[gray!20] (16,0) -- (16,20) -- (20,20) -- (40,0) -- cycle;

\draw[<->] (45,0) -- (0,0) -- (0,45);
\draw (40,0) -- (0,40);

\draw[very thick] (3,0) -- (3,3) -- (9,3) -- (9,7) -- (11,7) -- (11,13) -- (16,13);

\draw[very thick,dashed] (16,20) -- (20,20);
\draw[very thick] (16,13) -- (16,20);

\node[scale=1.05] at (3,3) {$\star$};
\node[scale=1.05] at (9,7) {$\star$};
\node[scale=1.05] at (11,13) {$\star$};

\draw[thick,fill=white] (16,20) circle[radius=0.32];
\node[anchor=south east] at (15.3,20.7) {$(16,\infty)$};



{\footnotesize
\node at (0,-1) {$0$};
\node at (1,-1) {$1$};
\node at (2,-1) {$2$};
\node at (3,-1) {$3$};
\node at (4,-1) {$4$};
\node at (9,-1) {$9$};
\node at (11,-1) {$11$};
\node at (16,-1) {$16$};
\node at (20,-1) {$\infty$};
\node at (40,-1) {$0^c$};
\node at (39,-1) {$1^c$};
\node at (38,-1) {$2^c$};
\node at (37,-1) {$3^c$};
\node at (36,-1) {$4^c$};
\node at (10,-1) {$\cdots$};
\node at (18.2,-1) {$\cdots$};
\node at (30,-1) {$\cdots$};

\node at (-1,0) {$0$};
\node at (-1,1) {$1$};
\node at (-1,2) {$2$};
\node at (-1,3) {$3$};
\node at (-1,4) {$4$};
\node at (-1,7) {$7$};
\node at (-1,13) {$13$};
\node at (-1,16) {$16$};
\node at (-1,20) {$\infty$};
\node at (-1,40) {$0^c$};
\node at (-1,39) {$1^c$};
\node at (-1,38) {$2^c$};
\node at (-1,37) {$3^c$};
\node at (-1,36) {$4^c$};
\node at (-1,10) {$\vdots$};
\node at (-1,18) {$\vdots$};
\node at (-1,30) {$\vdots$};
}

\draw[dotted]  (16,20) -- (0,20);
\draw[dotted] (20,0)  -- (20,20);

\node at (44,-2) {$T$};
\node at (-2,44) {$T$};
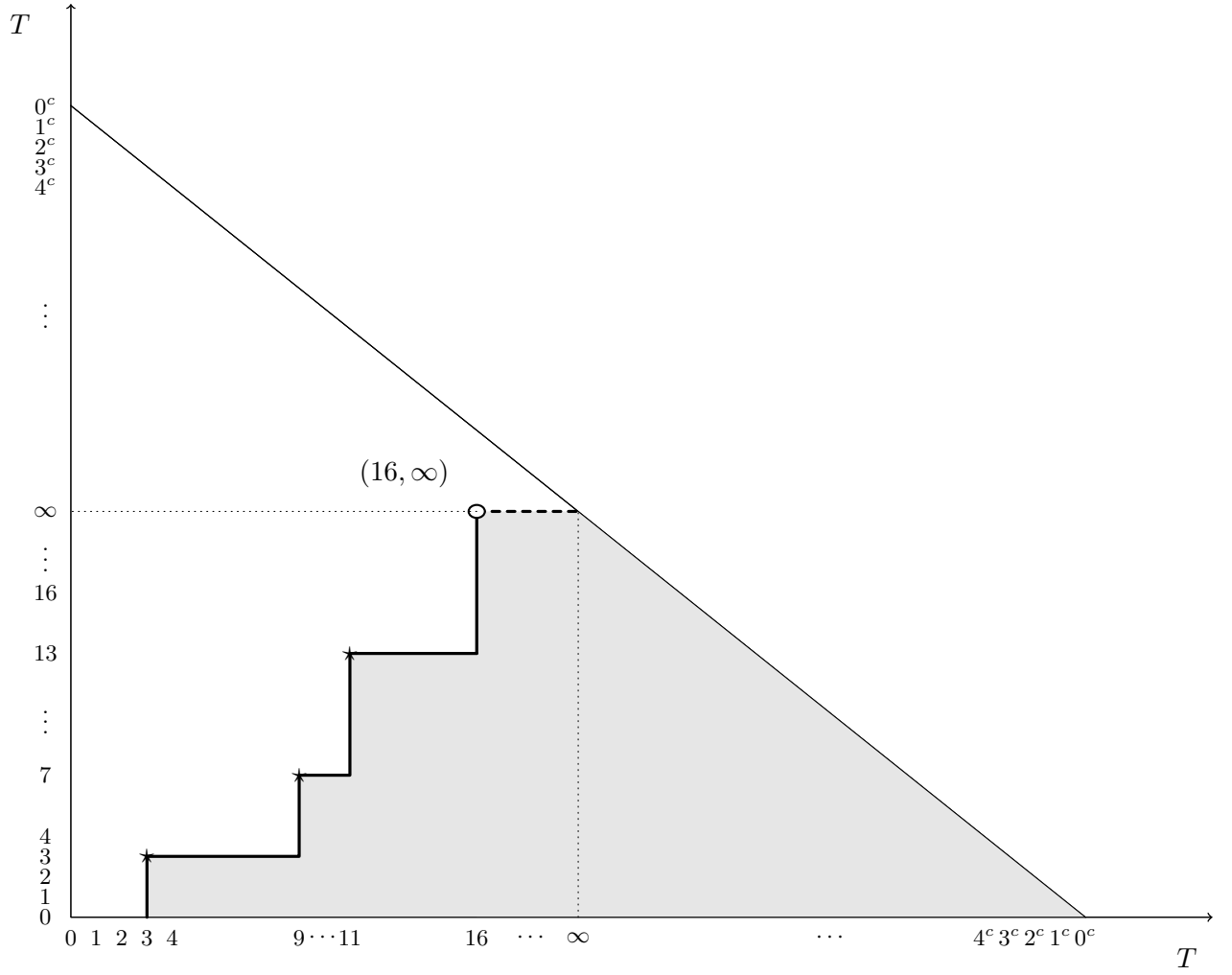
\node at (20,-4.2) {Figure 4: Shape form II prototype; $S=\uparrow\{(3,3),(9,7),(11,13)\}\cup J_{16,\infty}$};

\end{tikzpicture}
\end{center}

{\color{black}
Here $a$ wins under any of the three support-and-opposition conditions discussed after Definition \ref{61}, or whenever at least 16 voters support $a$ and only finitely many oppose it. The open point marks the excluded boundary at which opposition becomes infinite.}

\begin{center}
\begin{tikzpicture}[xscale=0.35, yscale=0.28, line cap=round, line join=round]

\draw[<->] (45,0) -- (0,0) -- (0,45);
\draw (40,0) -- (0,40);

\fill[white] (40,40) -- (40,0) -- (0,40) -- cycle;
\draw (40,0) -- (0,40);

\fill[gray!20] (3,0) -- (3,3) -- (9,3) -- (9,7) -- (11,7) -- (11,13)
  -- (27,13) -- (40,0) -- cycle;
\fill[gray!20] (20,0) -- (20,18) -- (22,18) -- (40,0) -- cycle;

\draw[<->] (45,0) -- (0,0) -- (0,45);
\draw (40,0) -- (0,40);

\draw[very thick] (3,0) -- (3,3) -- (9,3) -- (9,7) -- (11,7) -- (11,13) -- (20,13);

\draw[very thick,dashed] (20,13) -- (20,18);
\draw[very thick] (20,18) -- (22,18);

\node[scale=1.05] at (3,3) {$\star$};
\node[scale=1.05] at (9,7) {$\star$};
\node[scale=1.05] at (11,13) {$\star$};

\draw[thick,fill=white] (20,18) circle[radius=0.32];
\node[anchor=west] at (15.5,18.8) {$(\infty,18)$};


{\footnotesize
\node at (0,-1) {$0$};
\node at (1,-1) {$1$};
\node at (2,-1) {$2$};
\node at (3,-1) {$3$};
\node at (4,-1) {$4$};
\node at (9,-1) {$9$};
\node at (11,-1) {$11$};
\node at (20,-1) {$\infty$};
\node at (40,-1) {$0^c$};
\node at (39,-1) {$1^c$};
\node at (38,-1) {$2^c$};
\node at (37,-1) {$3^c$};
\node at (36,-1) {$4^c$};
\node at (10,-1) {$\cdots$};
\node at (18.2,-1) {$\cdots$};
\node at (30,-1) {$\cdots$};
\node at (22.3,-1) {$\cdots$};

\node at (-1,0) {$0$};
\node at (-1,1) {$1$};
\node at (-1,2) {$2$};
\node at (-1,3) {$3$};
\node at (-1,4) {$4$};
\node at (-1,7) {$7$};
\node at (-1,13) {$13$};
\node at (-1,18) {$18$};
\node at (-1,20) {$\infty$};
\node at (-1,40) {$0^c$};
\node at (-1,39) {$1^c$};
\node at (-1,38) {$2^c$};
\node at (-1,37) {$3^c$};
\node at (-1,36) {$4^c$};
\node at (-1,10) {$\vdots$};
\node at (-1,16) {$\vdots$};
\node at (-1,22) {$\vdots$};
\node at (-1,30) {$\vdots$};
}

\draw[dotted] (20,0) -- (20,13);\draw[dotted] (20,18)--(20,20)--(0,20);

\node at (44,-2) {$T$};
\node at (-2,44) {$T$};
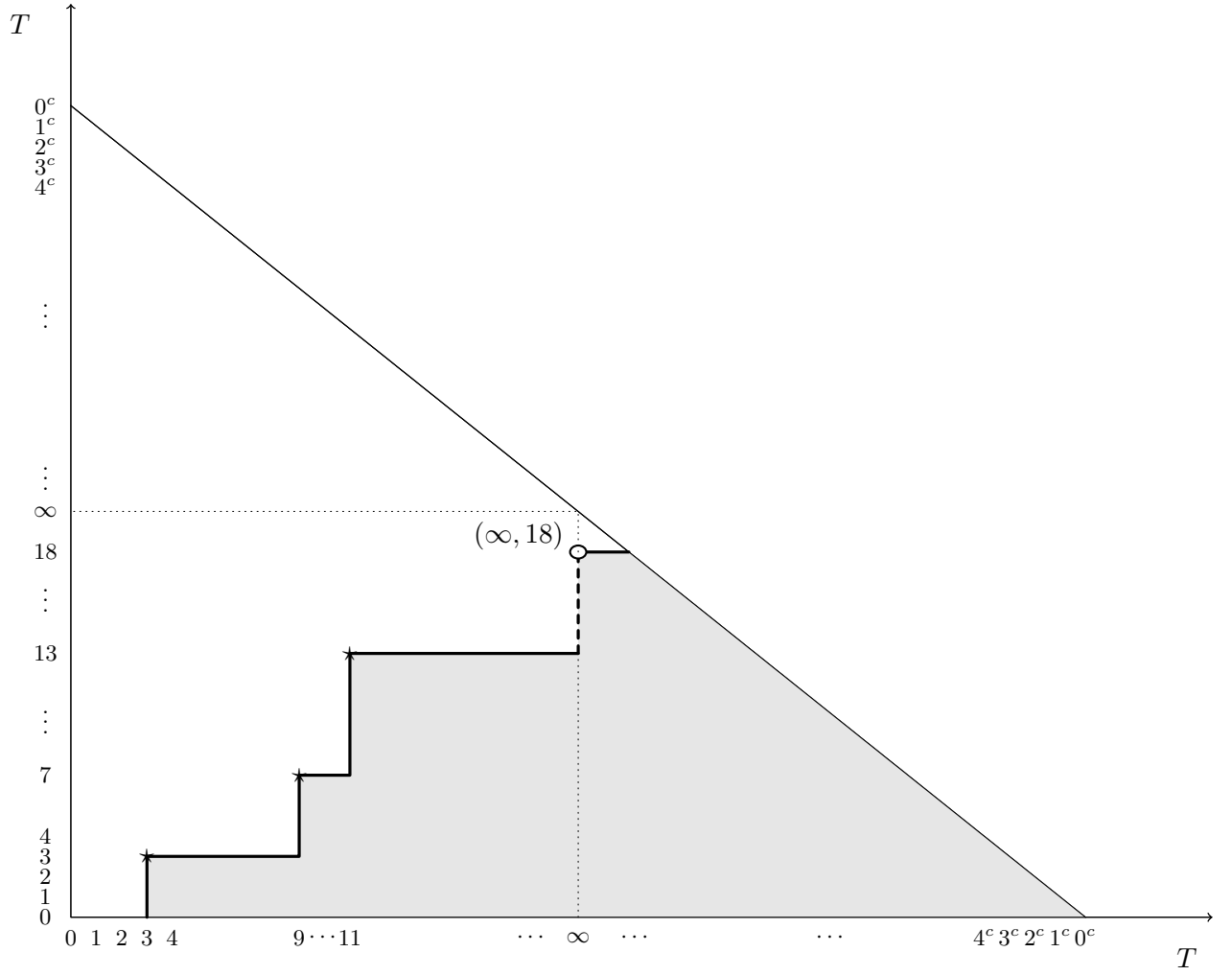
\node at (20,-4.2) {Figure 5: Shape form II prototype; $S=\uparrow\{(3,3),(9,7),(11,13)\} \cup K_{\infty,18}$};

\end{tikzpicture}
\end{center}

{\color{black}In Figure 5, $a$ is chosen if and only if one of the first three conditions described for Figure 4 holds, or all but finitely many voters support $a$ and at most 18 oppose it. This additional condition, represented by $K_{\infty,18}$, requires the indifferent voters to be finite as well. For example, a profile with infinitely many supporters, 18 opponents, and infinitely many indifferent voters yields $a$ under the rule in Figure 4 but $b$ under the rule in Figure 5.
}

\begin{center}
\begin{tikzpicture}[xscale=0.35, yscale=0.28, line cap=round, line join=round]

%

\draw[<->] (45,0) -- (0,0) -- (0,45);
\draw (40,0) -- (0,40);

\fill[white] (40,40) -- (40,0) -- (0,40) -- cycle;
\draw (40,0) -- (0,40);

\fill[gray!20]
  (3,0) -- (3,3) -- (9,3) -- (9,7) -- (10,7) -- (10,20)
  -- (14,20) -- (14,23) -- (17,23) -- (40,0) -- cycle;

\draw[<->] (45,0) -- (0,0) -- (0,45);
\draw (40,0) -- (0,40);

\draw[very thick]
  (3,0) -- (3,3) -- (9,3) -- (9,7) -- (10,7) -- (10,20);
\draw[very thick] (14,20) -- (14,23) -- (17,23) -- (40,0);

\draw[very thick,dashed] (10,20) -- (14,20);

\node[scale=1.05] at (3,3) {$\star$};
\node[scale=1.05] at (9,7) {$\star$};
\node[scale=1.05] at (14,23) {$\star$};

\draw[thick,fill=white] (10,20) circle[radius=0.32];
\node[anchor=south east] at (9.4,20.7) {$(10,\infty)$};



{\footnotesize
\node at (0,-1) {$0$};
\node at (1,-1) {$1$};
\node at (2,-1) {$2$};
\node at (3,-1) {$3$};
\node at (4,-1) {$4$};
\node at (9,-1) {$9$};
\node at (10,-1) {$10$};
\node at (17,-1) {$17$};
\node at (14,-1) {$14$};

\node at (20,-1) {$\infty$};
\node at (40,-1) {$0^c$};
\node at (39,-1) {$1^c$};
\node at (38,-1) {$2^c$};
\node at (37,-1) {$3^c$};
\node at (36,-1) {$4^c$};
\node at (23,-1.15) {$17^c$};
\node at (12.8,-1) {$\cdots$};
\node at (18.2,-1) {$\cdots$};
\node at (30,-1) {$\cdots$};

\node at (-1,0) {$0$};
\node at (-1,1) {$1$};
\node at (-1,2) {$2$};
\node at (-1,3) {$3$};
\node at (-1,4) {$4$};
\node at (-1,7) {$7$};
\node at (-1,10) {$10$};
\node at (-1,20) {$\infty$};
\node at (-1,23) {$17^c$};
\node at (-1,40) {$0^c$};
\node at (-1,39) {$1^c$};
\node at (-1,38) {$2^c$};
\node at (-1,37) {$3^c$};
\node at (-1,36) {$4^c$};
\node at (-1,13.5) {$\vdots$};
\node at (-1,18) {$\vdots$};
\node at (-1,30) {$\vdots$};
}

\draw[dotted] (14,20) -- (20,20);
\draw[dotted] (11,20) -- (0,20);
\draw[dotted] (20,0) -- (20,20);

\node at (44,-2) {$T$};
\node at (-2,44) {$T$};
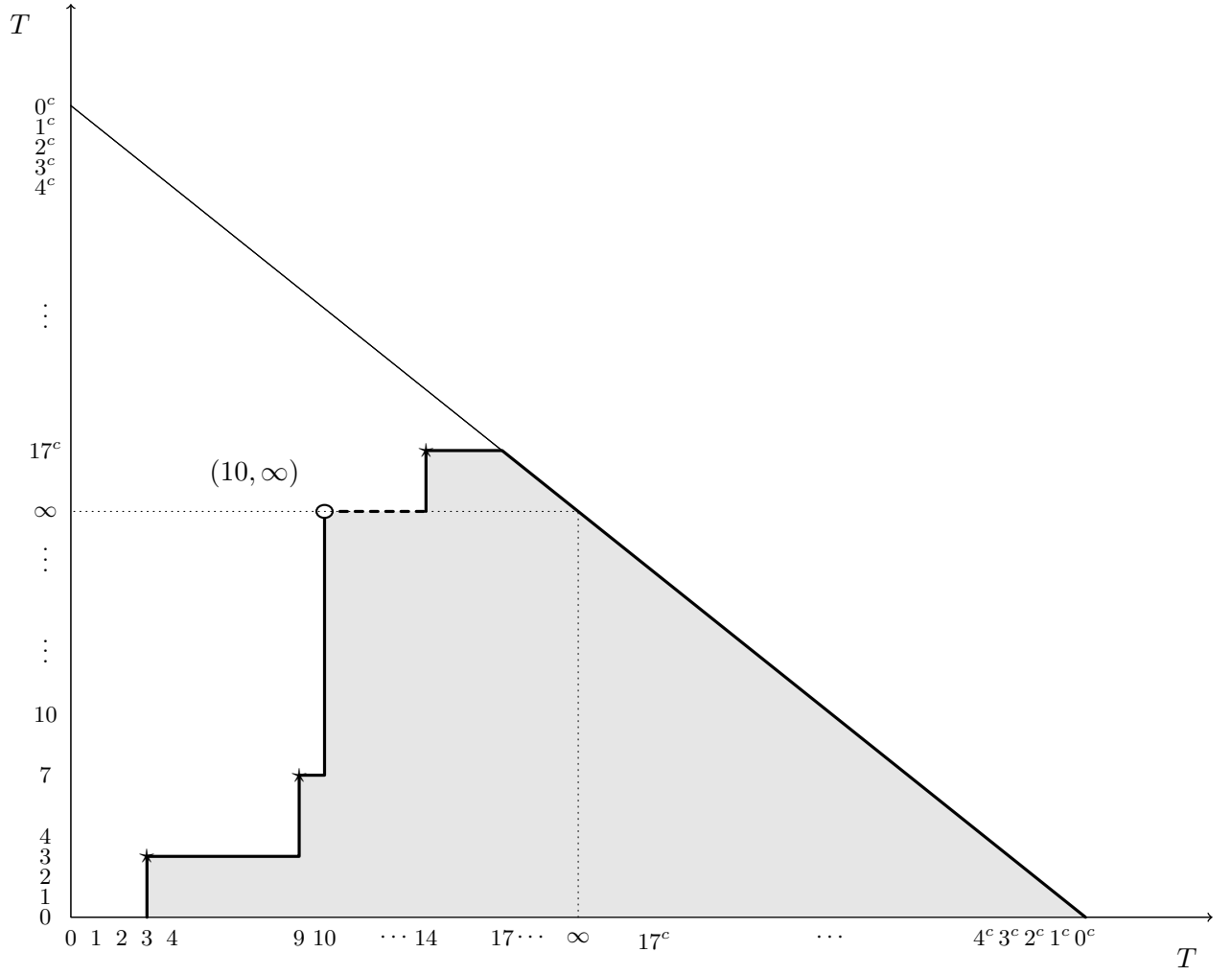
\node[align=center] at (20,-4.2) {Figure 6: Shape form III prototype; $S=\uparrow\{(3,3),(9,7),(14,17^c)\}\cup J_{10,\infty}$};

\end{tikzpicture}
\end{center}
{\color{black}In Figure 6, $a$ is chosen if and only if there are at least 3 supporters and at most 3 opponents; or at least 9 supporters and at most 7 opponents; or at least 14 supporters and at least 17 voters in total who either support $a$ or are indifferent; or at least 10 supporters and only finitely many opponents. The third condition translates the corner $(14,17^c)$: its second coordinate places a lower bound on the number of voters who do not strictly oppose $a$. The last condition is supplied by $J_{10,\infty}$.
}


\begin{center}
\begin{tikzpicture}[xscale=0.35, yscale=0.28, line cap=round, line join=round]


\draw[<->] (45,0) -- (0,0) -- (0,45);
\draw (40,0) -- (0,40);

\fill[white] (40,40) -- (40,0) -- (0,40) -- cycle;
\draw (40,0) -- (0,40);

\fill[gray!20]
  (3,0) -- (3,2) -- (9,2) -- (9,4) -- (20,4)
  -- (20,12) -- (23,12) -- (23,14) -- (25,14)
  -- (25,15) -- (40,0) -- cycle;

\draw[<->] (45,0) -- (0,0) -- (0,45);
\draw (40,0) -- (0,40);

\draw[very thick]
  (3,0) -- (3,2) -- (9,2) -- (9,4) -- (20,4) -- (20,6);
\draw[very thick,dashed] (20,6) -- (20,12);
\draw[very thick]
  (20,12) -- (23,12) -- (23,14) -- (25,14)
  -- (25,15) -- (40,0);

\node[scale=1.05] at (3,2) {$\star$};
\node[scale=1.05] at (9,4) {$\star$};
\node[scale=1.05] at (20,6) {$\star$};
\node[scale=1.05] at (23,14) {$\star$};
\node[scale=1.05] at (25,15) {$\star$};

\draw[thick,fill=white] (20,12) circle[radius=0.32];
\node[anchor=west,scale=0.9] at (15,12) {$(\infty,12)$};


\node[anchor=east,scale=0.9] at (19.5,5.7) {$(\infty,6)$};

{\footnotesize
\node at (0,-1) {$0$};
\node at (1,-1) {$1$};
\node at (2,-1) {$2$};
\node at (3,-1) {$3$};
\node at (4,-1) {$4$};
\node at (9,-1) {$9$};
\node at (20,-1) {$\infty$};
\node at (23,-1) {$17^c$};
\node at (25,-1) {$15^c$};
\node at (40,-1) {$0^c$};
\node at (39,-1) {$1^c$};
\node at (38,-1) {$2^c$};
\node at (37,-1) {$3^c$};
\node at (36,-1) {$4^c$};
\node at (10,-1) {$\cdots$};
\node at (17.8,-1) {$\cdots$};
\node at (30,-1) {$\cdots$};
\node at (34,-1) {$\cdots$};

\node at (-1,0) {$0$};
\node at (-1,1) {$1$};
\node at (-1,2) {$2$};
\node at (-1,3) {$3$};
\node at (-1,4) {$4$};
\node at (-1,6) {$6$};
\node at (-1,12) {$12$};
\node at (-1,20) {$\infty$};
\node at (-1,40) {$0^c$};
\node at (-1,39) {$1^c$};
\node at (-1,38) {$2^c$};
\node at (-1,37) {$3^c$};
\node at (-1,36) {$4^c$};
\node at (-1,8.8) {$\vdots$};
\node at (-1,23) {$\vdots$};
\node at (-1,30) {$\vdots$};
}

\draw[dotted] (20,0) -- (20,20)-- (0,20);
\draw[dotted] (23,0) -- (23,14) -- (0,14);
\draw[dotted] (25,0) -- (25,15) -- (0,15);

\node at (44,-2) {$T$};
\node at (-2,44) {$T$};
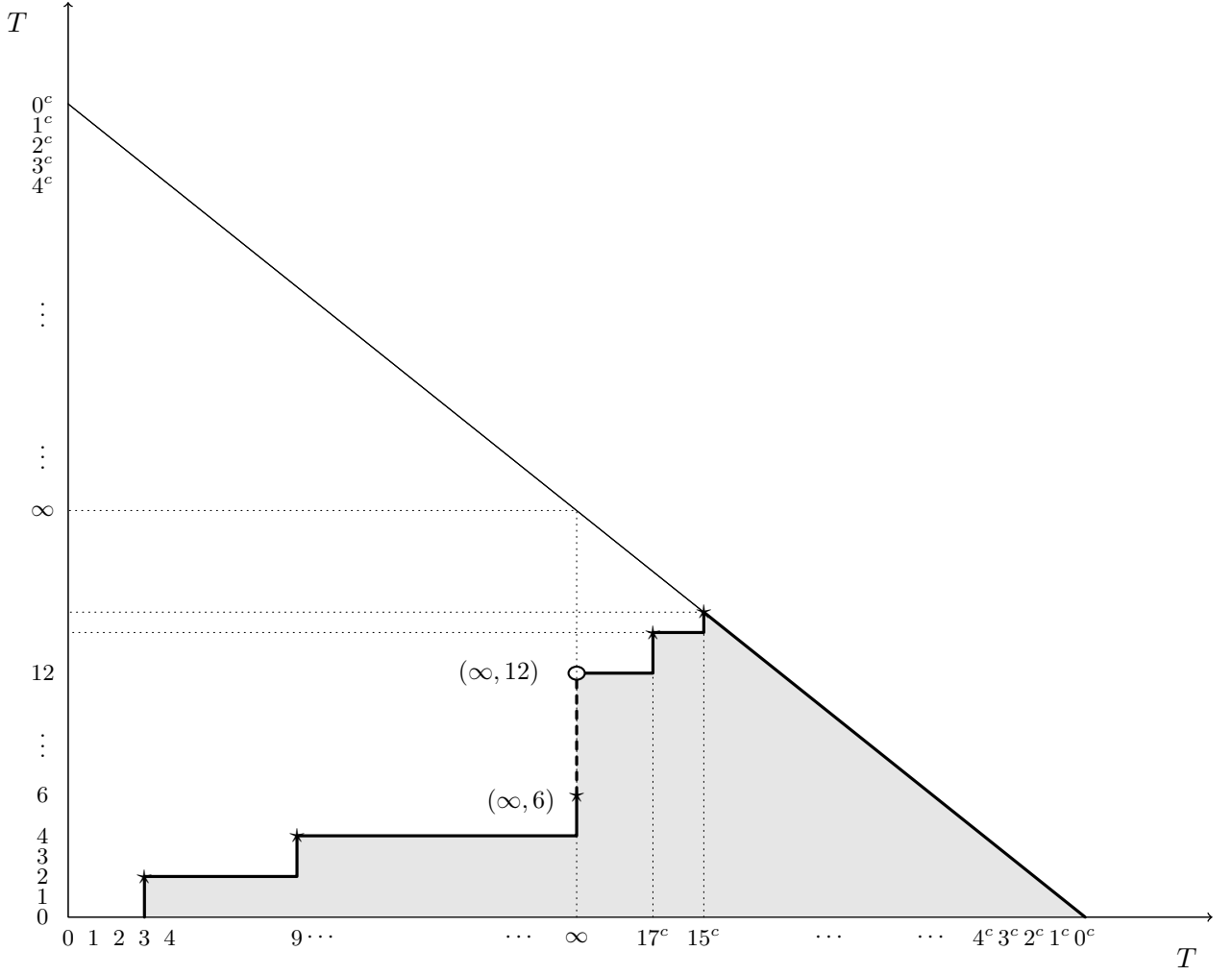
\node at (20,-4.2) {Figure 7: Shape form IV prototype; $S=\uparrow \{(3,2), (9,4), (\infty,6), (17^c,14), (15^c,15)\}\cup K_{\infty,12}$};

\end{tikzpicture}
\end{center}
{\color{black}In Figure 7, $a$ wins with at least 3 supporters and at most 2 opponents; with at least 9 supporters and at most 4 opponents; or with infinitely many supporters and at most 6 opponents. It also wins if all but finitely many voters support $a$ and at most 12 oppose it; if all but at most 17 voters support $a$ and at most 14 oppose it; or if all but at most 15 voters support $a$. In every other case, $b$ is chosen. The condition involving 12 opponents comes from $K_{\infty,12}$, while the bounds of 17 and 15 count opponents and indifferent voters together.
}

\section{Conclusions}\label{conclusioni}
{\color{black}
We have extended the geometric approach to anonymous strategy-proof
binary social choice developed for finite electorates in Basile, Rao and
Bhaskara Rao \cite{mss} to a countably infinite society, taking coalitional
strategy-proofness --- rather than its individual counterpart --- as the
relevant non-manipulability requirement, and admitting indifference in
individual preferences.

\smallskip
\noindent
The extension required two separate steps, each summarized by a main
result. First, Theorem~5.3 shows that anonymous, coalition strategy-proof
binary social choice functions on the universal weak domain are in
one-to-one correspondence with the super order closed subsets of an
infinite triangular grid $G$ built from the linearly ordered set $T$ of cardinality
signatures. This is the exact infinite-society counterpart of the finite
representation via SOC subsets of $G_n$, and it rests on the permutation
lemma (Lemma~5.4) --- the genuinely infinitary fact that makes the
passage from signature comparisons to actual coalition comparisons
possible. Second, Corollary \ref{thm:zigzag} shows that every nonempty SOC subset of
$G$ falls into one of four shapes (I--IV), determined by a possibly
infinite antichain of genuine corners together with at most one of two
missing-boundary regions, $J_{k,\infty}$ and $K_{\infty,\ell}$, that have
no counterpart in the finite grid $G_n$. Example \ref{62} shows why these
extra regions are needed: it exhibits a SOC subset of $G$ that cannot be
generated by any antichain at all, Proposition \ref{use of parameters} further shows that every SOC subset of $G$ corresponds to a unique admissible
pair $(A,R)$, where $A$ is its antichain of minimal elements and $R$ is either empty or an essential
$J$- or $K$-region. Together with Theorem \ref{53}, this gives a one-to-one parametrization of all
anonymous coalition strategy-proof rules.

\smallskip
\noindent
We regard the failure of generation
by minimal elements as more than a technical curiosity. In case of anonymity, it is the
strategy-proof, binary-choice counterpart of a pattern long documented in
the sensu lato Arrovian  tradition: passing to an infinite society does not simply
add more of the same finite structures, it enlarges the admissible class
of collective choice rules in ways finite reasoning cannot anticipate ---
there, through non-dictatorial rules generated by free ultrafilters; 
here, through rules whose defining regions may require missing-boundary components in
addition to the upward closure of their minimal elements.
 A phenomenon  that has no finite-society correspondent.

\smallskip
\noindent
{\color{black}Coalitional strategy-proofness is not adopted here merely as standard terminology: on the grid $G$, a single voter's report can move a coalition's cardinality signature by at most one step \emph{within} the finite band $\{0,1,2,\ldots\}$ or within the cofinite band $\{\ldots,2^{c},1^{c},0^{c}\}$, but no individual deviation can cross between the finite, $\infty$, and cofinite bands. The comparisons across these bands involving the vertex $(\infty,\infty)$ and the missing-boundary corners are therefore invisible to individual manipulation. Coalitional strategy-proofness constrains these comparisons by also ruling out profitable deviations by infinite coalitions. In this precise sense, coalitional strategy-proofness is the notion under which the classification of Section \ref{sei} is valid as a complete characterization of the admissible rules.
}}

 \section{Appendix: proofs}\label{dimostrazioni}

 \subsection{ }\label{dim classificazione}
 \begin{proof}[Proof of the classification Theorem \ref{classificazione}]
If $\F=\varnothing$, then $\F=\G_{\co-0}$. Suppose now that $\F\neq\varnothing$.

Assume first that $\F$ contains a finite set. Let
\[
k_0=\min\{|E|:E\in\F,\ E\text{ finite}\}.
\]
We prove that $\F=\G_{k_0}$. If $E\in\F$, then by the definition of $k_0$ and by SOC, $|E|\geq k_0$. Conversely, let $E\subseteq\N$ with $|E|\geq k_0$. Take $F\in\F$ with $|F|=k_0$. There is a permutation $\pi$ of $\N$ such that $\pi(F)\subseteq E$. By anonymity, $\pi(F)\in\F$; by SOC, $E\in\F$. Hence $\F=\G_{k_0}$.

Assume now that the elements of $\F$ are all infinite. There are two cases.

First, suppose that $\F$ contains an infinite set $F$ which is not cofinite. We claim that $\F=\G_\infty$. Let $E$ be infinite. If $E$ is not cofinite, then $F$ and $E$ are both countably infinite with countably infinite complements; hence a permutation $\pi$ of $\N$ exists with $\pi(F)=E$. Therefore $E\in\F$. If $E$ is cofinite, take an infinite subset $E_0\subseteq E$ whose complement is infinite. By the previous argument $E_0\in\F$, and then SOC gives $E\in\F$. Thus every infinite set belongs to $\F$.

Second, suppose that every element of $\F$ is cofinite. Set
\[
K=\{|F^c|:F\in\F\}\subseteq \N_0.
\]
If $K$ has a maximum, say $k_0$, we prove that
\[
\F=\{E\subseteq\N: |E^c|\leq k_0\}=\G_{\co-(k_0+1)}.
\]
The inclusion $\F\subseteq\{E:|E^c|\leq k_0\}$ is immediate. Conversely, take $E$ with $|E^c|\leq k_0$ and take $F\in\F$ with $|F^c|=k_0$. There is a permutation $\pi$ of $\N$ such that $\pi(E^c)\subseteq F^c$. Equivalently, $F\subseteq \pi(E)$. By SOC, $\pi(E)\in\F$, and then by anonymity $E\in\F$.

If $K$ has no maximum, we prove that $\F=\G_{\co-\infty}$. Let $E$ be cofinite. Since $K$ has no maximum, there is $F\in\F$ such that $|F^c|>|E^c|$. By the argument just used, $E\in\F$. Thus every cofinite set belongs to $\F$.

The proof is complete.
\end{proof}

 \subsection{ }\label{dim 44}
 \begin{proof}[Proof of the canonical representation Theorem \ref{44}]
 This is a straightforward consequence of Corollary \ref{prima rappresentazione}. 
 
\noindent The latter gives that the elements of 
$\SCF_{AN}^s$  are 
$\phi_{\G_0}, \dots, \phi_{\G_k}, \dots, \phi_{\G_\infty}, \phi_{\G_{co-\infty}}, \dots,   \phi_{\G_{co-k}},\dots,  \phi_{\G_{co-0}}.$ 

Now corresponding to every case of the anonymous committees $\G$'s we provide the SOC subset $C$ of $(T, \le_T)$ such that
$$D(a, P)\in \G\Leftrightarrow \rchi_s(P)\in C.$$ We list next the correspondences.

- For $\phi_{\G_0}$ the SOC $C$ is $T$;

- for $\phi_{\G_k}$ the SOC $C$ is the interval $[k, 0^c]$ of $T$;

- for $\phi_{\G_\infty}$ the SOC $C$ is the interval $[\infty, 0^c]$ of $T$;

- for $\phi_{\G_{co-\infty}}$ the SOC $C$ is the interval $]\infty, 0^c]$ of $T$;

- for $\phi_{\G_{co-k}}$ the SOC $C$ is the interval $[(k-1)^c, 0^c]$ of $T$; \quad and finally

- for $\phi_{\G_{co-0}}$ the SOC $C$ is the empty set.

\noindent As we exhaust all of $\SCF_{AN}^s$, we have to remark that with the corresponding $C$'s we also exhaust all SOC subsets of $T$. Indeed,
 {\color{black} excluding $T$ and the empty set, the remaining SOC sets $C$ are such that 
$\O\subset C\subset T$, and of course there are only two possibilities: either $C$ has the minimum ($>0$) or $C$ does not. In the first case one can see that $C$ is one of the intervals $[k, 0^c]$ $[\infty, 0^c]$ $[k^c, 0^c]$. Whereas in the second case one can see that $C=]\infty, 0^c]$.}
\end{proof}

\subsection{ }\label{dim permutation lemma}
\begin{proof}[Proof of the permutation Lemma \ref{permutation lemma}]
Given the sets $A_1,B_1,A_2,B_2\subseteq\N$ with
\[
A_1\cap B_1=\varnothing,
\quad
A_2\cap B_2=\varnothing, \quad
\siga(A_1)\leq_T\siga(A_2),
\mbox{ and }
\siga(B_1)\geq_T\siga(B_2),
\]
we prove that there is a permutation $\pi$ of $\N$ such that
\[
\pi(A_1)\subseteq A_2,
\qquad
B_2\subseteq \pi(B_1).
\]
First we observe that, as it can be promptly verified,
$$ T\ni \alpha\le_T \sigma(F), \, F\subseteq \N \Rightarrow \, \exists E\subseteq F\mbox{ such that } \sigma(E)=\alpha.$$
Due to this,  we can prove the assertion under the stronger assumption that
$$(\star)\qquad \sigma(A_1)= \sigma(A_2), \mbox{ and }\sigma(B_2)= \sigma(B_1),$$ without loss of generality.

\noindent
Indeed, given our assumptions we can take 

-- $D_2\subseteq A_2$ with $\sigma(D_2)=\sigma (A_1)$, and

-- $D_1\subseteq B_1$ with $\sigma(D_1)=\sigma (B_2)$.

\noindent Then, if  the lemma under the signature-equality assumption has been proved, we can apply it to the pairs $A_1, D_1$ and $D_2, B_2$, obtaining a permutation $\pi$ of $\N$ with 
$\pi(A_1)\subseteq D_2$, and $ B_2\subseteq \pi(D_1)$. This is clearly enough.

\bigskip\noindent
So, let us assume  $(\star)$. We consider the following four cases.

\medskip
(case 1): $\sigma(A_1)\in]\infty, 0^c]$ and (necessarily) $\sigma(B_1)$ finite; say
we can set  $\sigma(A_i)=k^c, \sigma(B_i)=m, |A_i^c|=k$, for some $k\ge m\in \N_0$. 

\medskip
(case 2): $\sigma(A_1)=\infty$ and  $\sigma(B_1)$ finite; say
we can set  $\sigma(A_i)=\infty, \sigma(B_i)=m$ for some $m\in \N_0$, with the sets $ A_i, A_i^c$ both infinite. 

\medskip
(case 3): $\sigma(A_1)=\infty=\sigma(B_1)$;\footnote{When $\sigma(A_1)=\infty$ the possibility that $\sigma(B_1)$ is some $k^c$ does not arise.}

\medskip
(case 4): both $\sigma(A_1)$ and $\sigma(B_1)$ are finite; say
we can set  $\sigma(A_i)=k, \sigma(B_i)=m, |A_i^c|=\infty$, for some $k,m\in \N_0$. 

\medskip\noindent
In principle, one should discuss even two other cases:

-- $\sigma(A_1)$ finite,  $\sigma(B_1)=\infty$; and 

-- $\sigma(A_1)$ finite,  $\sigma(B_1)\in]\infty, 0^c]$. 

\noindent
However, by symmetry, they are \lq\lq identical\rq\rq\, respectively to (case 2) and (case 1). Indeed it is enough to replace 
in order the pairs $A_1, B_1$, $A_2, B_2$ with the pairs $B_2, A_2$, $B_1, A_1$.\footnote{\color{black} Once the permutation $\pi$ is consequently obtained, this has to be replaced with its inverse.} No other cases arise.

\noindent
Finding the desired permutation in the three cases 1, 2, 4 is straightforward, 
and one gets $\pi(A_1)=A_2$, and $ B_2=\pi(B_1)$.

\noindent
In (case 3), the sets $A_i, B_i$ are infinite whereas we do not know the cardinalities of the sets $A_i^c\setminus B_i$. We can assume without loss of generality that $|A_1^c\setminus B_1|<|A_2^c\setminus B_2|$. Then the set $A_1^c\setminus B_1$ is finite. We can partition $A_2^c\setminus B_2$ in two  subsets $I_2$ and $J_2$ with $I_2$ having the same cardinality as $A_1^c\setminus B_1$. Hence we can get the permutation $\pi$ by gluing a bijection of $A_1$ onto $A_2$, a bijection of $A_1^c\setminus B_1$ onto $I_2$, and a bijection of $B_1$ onto $B_2\cup J_2$.
\footnote{\color{black}Note that even with signature-equality assumption it is not possible to get $\pi(A_1)=A_2$, and $ B_2=\pi(B_1)$.}
\end{proof}

\subsection{ }\label{dim prop:infinite-antichains}
\begin{proof}[Proof of Proposition \ref{prop:infinite-antichains} on infinite antichains] 

We show that if we remove, possibly, the point $(\infty, \infty)$ from an infinite antichain of 
$G$, the residual antichain, which we denote by 
$A$, turns out to be a set like
 $\{(k_1, m_1), (k_2, m_2), \dots \}$ where the $k$'s and the $m$'s form two strictly increasing sequences in $\N_0$. 

Observe that  {\it every nonempty section of   an antichain of $G$ is a singleton.} Now we aim to prove that

(CLAIM): $A\subseteq \N_0\times\N_0$.

First we observe that $A\subseteq [0, \infty]\times [0, \infty]$. For, {\it if $(\alpha_0, \beta_0)\in A$ and $\alpha_0=k_0^c$, we get that $A$ is finite}.
Indeed, the construction of the poset $G$ tells us that $\beta_0\in [0,k_0]$. Moreover, 

$\big[(\alpha_0, \beta_0)\neq(\alpha, \beta)\in A\big]\Rightarrow \big[\beta_0\neq\beta,$ and $\beta\in [0, k_0]\big]$, what can be written as 
$A\setminus \{(\alpha_0, \beta_0)\}\subseteq\bigcup_{\beta\in[0, k_0], \beta\neq\beta_0} A^\beta$, so the finiteness of $A$.
{\it Similarly if for the second coordinate we have $\infty<_T\beta_0.$}

The sections $A_\infty, A^\infty$ are at most singletons, hence infinitely many points of $A$ belong $\N_0\times\N_0$. If one of the two sections is nonempty it will contain a point which is comparable with elements of $A\cap\N_0\times\N_0$. So the claim is proved.

Given that the (CLAIM) is true to obtain the assertion of the proposition is elementary.\end{proof}

\subsection{ }\label{dim lem:covering}
\begin{proof}[\color{black}Proof of  the covering lemma \ref{lem:covering}]
Let \((\bar\alpha,\bar \beta)\in S\).  For the horizontal section $S^{\bar\beta}$ 
we have the following three cases to analyze.
\begin{itemize}
\item[ ] Case 1: The set $S^{\bar\beta}$ has minimum $\alpha_0$ and $S_{\alpha_0}$ has maximum $\beta_0$.
\item[ ] Case 2: The set $S^{\bar\beta}$ has minimum $\alpha_0$ but $S_{\alpha_0}$ does not admit maximum.
\item[] Case 3: The set $S^{\bar\beta}$ does not admit minimum.
\end{itemize}
In each of the cases we exhibit a corner that absorbs $(\bar\alpha, \bar\beta)$ as asserted. Precisely, 

\begin{itemize}
\item [] Case 1: we show that a corner of S, of type 1, that absorbs  $(\bar\alpha, \bar\beta)$ is $(\alpha_0, \beta_0)$.
\item [] Case 2: we show that a corner of S, of type 2, that absorbs  $(\bar\alpha, \bar\beta)$ is $(\alpha_0, \infty)$.
\item []  Case 3: we show that $(\bar\alpha, \bar\beta)$ is absorbed by a corner of S of type 3.
\end{itemize}
So, let us discuss the first case, where by definitions clearly \((\alpha_0,\bar\beta),\,(\alpha_0,\beta_0)\in S\), and \(\beta_0\ge_T \bar\beta\).

\noindent  Moreover we show, repeatedly using super order closedness of $S$, that $(\alpha_0,\beta_0)$  is a minimal point of $S$:
$$(\alpha_0,\beta_0)\ge_G (\gamma, \eta)\in S\Rightarrow (\alpha_0,\beta_0)=(\gamma, \eta).$$
Observe that 
$(\gamma,\beta_0)\ge_G (\gamma, \eta)\in S\Rightarrow (\gamma,\beta_0)\in S$. On the other hand  from $\beta_0\ge_T \bar\beta\) we obtain that $\gamma\in S^{\bar\beta}$. This means that $\gamma=\alpha_0$.

\noindent
Still it is possible that  \(\eta>_T\beta_0\). However, observe that we have obtained 
 $(\alpha_0,\eta)\in S$. Then \(\eta\in S_{\alpha_0}$, hence $\eta=\beta_0$.

\noindent Therefore \((\alpha_0,\beta_0)\) is a corner of type 1. Since \(\alpha_0\le_T \bar\alpha\) and \(\beta_0\ge_T \bar\beta\), we have
$
(\alpha_0,\beta_0)\le_G(\bar\alpha,\bar\beta),
$
and therefore \((\bar\alpha,\bar\beta)\in I_{\alpha_0,\beta_0}\).

\bigskip
\noindent
We move to consider the second case. As we have observed in Remark \ref{due} necessarily $\sup S_{\alpha_0}=\infty$, hence by definition $(\alpha_0,\infty)\notin S$,  $(\alpha_0,\bar\beta)\in S$,  \(\bar\beta<_T\infty\), $\alpha_0\le_T \bar\alpha$. The only possibility is that  \(\alpha_0\in\Nb\), otherwise the section $S_{\alpha_0}$ would be  finite, against the fact that we have assumed it has no maximum.

\noindent
Since S is SOC, we get $S_{\alpha_0}=\mathbb N_0$, i.e.
$$
(\alpha_0,m)\in S
\,\,
\hbox{for every }m\in\Nzero,\,\,
\mbox{ and consequently }
J_{\alpha_0,\infty}\subseteq S.
$$
Hence \(\alpha_0\in A_2(S)\). This determines the existence of the corner $(k_2(S), \infty)$of type 2. Since 
\(\alpha_0\le_T \bar\alpha\) and \(\bar\beta<_T\infty\), we have \((\bar\alpha,\bar\beta)\in J_{\alpha_0,\infty}\subseteq J_{k_2(S),\infty}\).

\bigskip
\noindent In case 3, as we have observed in Remark \ref{due} necessarily $\inf S^{\bar\beta}=\infty$, hence by definition $(\infty,\bar\beta)\notin S$, \(\bar\alpha >_T\infty\).
 As a consequence \(\bar\beta<_T\infty\). 
 
 \noindent Clearly \((\bar\alpha,\bar\beta)\in K_{\infty,\bar\beta}\). We claim that  

\[
K_{\infty,\bar\beta}\subseteq S,
\]

 \noindent namely that
 \[
(\gamma,\beta')\in S
\qquad
\hbox{for every }\gamma>_T\infty,
\ \beta'\le_T \bar\beta, \mbox{ with } (\gamma,\beta')\in G.
\]
Since $\gamma$ is not a lower bound for the section $S^{\bar\beta}$, there exists \(\alpha\in S^{\bar\beta}\) with \(\alpha<_T\gamma\). Obviously:
$$S\ni(\alpha,\bar\beta)\le_G(\gamma,\beta')
$$
hence super order closedness  of $S$ gives  \((\gamma,\beta')\in S\).
We have so obtained that \(\bar\beta\in A_3(S)\). This determines the existence of the corner $(\infty,\ell_3(S))$of type 3 absorbing \((\bar\alpha,\bar\beta)$.
\end{proof}


\subsection{ }\label{dim thm:zigzag}
\begin{proof}[Proof of  Corollary \ref{thm:zigzag} about the shape of the SOC subsets of $G$]
By Theorem~\ref{prop:exact} and Remark \ref{interpretazione1}, every nonempty SOC subset \(S\) is generated exactly by its genuine corners together with at most one  \(J\)-region and at most one  \(K\)-region.

If \(\tau_1(S)\) is infinite, Proposition~\ref{prop:infinite-antichains} gives the two possible compositions of \(\tau_1(S)\). Moreover \(\up \tau_1(S)\) already contains \(J_{\infty,\infty}\) and \(K_{\infty,\infty}\). A region \(J_{k,\infty}\) with finite \(k\) cannot coexist with the infinite sequence of finite-finite type 1 corners, because it would put, above all sufficiently far corners of the sequence, points of \(S\) on the same vertical line, contradicting the type 1 condition. Thus all possible missing-boundary contributions are redundant, and \(S=\up \tau_1(S)\). This proves case I.

Assume now that \(\tau_1(S)\) is finite. Since any point of the north-west triangle \(E\) is comparable with any point of the south-east triangle \(F\), the antichain \(\tau_1(S)\) cannot meet both \(E\) and \(F\). Hence exactly one of the following alternatives holds: \(\tau_1(S)\subseteq\Nb^2\), or \(\tau_1(S)\cap E\ne\varnothing\), or \(\tau_1(S)\cap F\ne\varnothing\).

If \(\tau_1(S)\subseteq\Nb^2\), Theorem~\ref{prop:exact} gives \(\up \tau_1(S)\), possibly with a \(J\)-term, possibly with a \(K\)-term. If both are present, then
\[
K_{\infty,\ell}\subseteq J_{k,\infty}
\]
for every \(k,\ell\in\Nb\), so the \(K\)-term may be omitted. This gives case II.

If \(\tau_1(S)\cap E\ne\varnothing\), take \((p,q^{\cmark})\in \tau_1(S)\cap E\). Then, for every \(\ell\in\Nb\),
\[
K_{\infty,\ell}\subseteq I_{p,q^{\cmark}}\subseteq\up \tau_1(S).
\]
Thus a possible \(K\)-term is redundant. A \(J\)-term may remain, giving case III.

If \(\tau_1(S)\cap F\ne\varnothing\), take \((p^{\cmark},q)\in \tau_1(S)\cap F\). A \(J\)-term cannot occur. Indeed, if \(J_{k,\infty}\subseteq S\), then, since \(k\le_T\infty<_T(p+1)^{\cmark}<_T p^{\cmark}\), the feasible point \(((p+1)^{\cmark},q)\) belongs to \(J_{k,\infty}\subseteq S\). This point has the same ordinate as \((p^{\cmark},q)\) and lies strictly to its left, contradicting the horizontal minimality condition in the definition of a type 1 corner. Therefore only a possible \(K\)-term may remain, giving case IV.
\end{proof}

{\color{black}
\subsection{ }\label{dim use of parameters}
\begin{proof}[Proof of Proposition \ref{use of parameters}]Every member of \(\mathcal R\) is SOC, so \(S_{A,R}\) is SOC for every admissible pair. Every nonempty member of \(\mathcal R\) has no minimal elements. Furthermore, since \(R\) is SOC and \(A\cap R=\varnothing\), no point of \(R\) can lie below a point of \(A\). Together with the antichain property of \(A\), these observations show that the minimal elements of \(S_{A,R}\) are exactly the points of \(A\).

To prove surjectivity, let \(S\) be a nonempty SOC subset of \(G\) and set \(A=\tau_1(S)\). By Theorem 6.8 and Remark 6.9, \(S\) can be written as
$$
S=\big(\uparrow A\big)\cup R
$$
with \(R\in\mathcal R\): if both exceptional terms occur, the \(K\)-term is omitted because it is contained in the \(J\)-term. If the remaining exceptional region is contained in \(\uparrow A\), take \(R=\varnothing\). Otherwise, retain it. In the latter case, \(A\cap R=\varnothing\), because \(R\subseteq S\) has no minimal elements whereas every point of \(A\) is minimal in \(S\). Thus \((A,R)\) is admissible. The empty SOC set is obtained from \((\varnothing,\varnothing)\).

For injectivity, suppose that

$$
\big(\uparrow A\big)\cup R=\big(\uparrow A'\big)\cup R'
$$

for two admissible pairs. Since the minimal elements of the common set are both \(A\) and \(A'\), we have \(A=A'\). Write \(U=\uparrow A\).

If exactly one of \(R,R'\) were empty, the other would be contained in \(U\), contrary to admissibility. It remains to exclude the possibility that \(R\) and \(R'\) are distinct nonempty regions.

Suppose first that at least one is a \(J\)-region. Relabelling the two regions if necessary, we may take

$$
R=J_{k,\infty},
$$

where \(R'\) is either a \(K\)-region or \(J_{k',\infty}\) with \(k<k'\). For every \(m\in\mathbb N_0\),

$$
(k,m)\in R\setminus R',
$$

so equality of the two unions implies \((k,m)\in U\). Since

$$
\uparrow\{(k,m):m\in\mathbb N_0\}=J_{k,\infty},
$$

we obtain \(R\subseteq U\), contradicting admissibility.

Finally, suppose that

$$
R=K_{\infty,\ell},
\qquad
R'=K_{\infty,\ell'},
\qquad
\ell<\ell'.
$$

Here \(\ell\) is finite. Put \(q=\ell+1\), so \(q\leq\ell'\). For every integer \(p\geq q\),

$$
(p^c,q)\in R'\setminus R,
$$

and hence \((p^c,q)\in U\). Consequently,

$$
K_{\infty,q}
=
\uparrow\{(p^c,q):p\in\mathbb N_0,\ p\geq q\}
\subseteq U.
$$

But \(K_{\infty,\ell}\subseteq K_{\infty,q}\), again contradicting admissibility. Therefore \(R=R'\), completing the proof. 
\end{proof}
}

\end{document}